\documentclass[journal,twoside,web]{ieeecolor}
\usepackage{generic}
\usepackage{textcomp}

\makeatletter
\newcommand\Xhline[1]{\noalign{\ifnum0=`}\fi\hrule height #1\futurelet
  \reserved@a\@xhline}
\makeatother

\usepackage{xcolor}
\usepackage{graphicx}       
\usepackage{bm}
\usepackage{amssymb,amsmath,amsfonts}
\usepackage{subcaption}
\usepackage{color}
\usepackage{algpseudocode}
\usepackage{cases}
\usepackage{multirow}
\definecolor{subsectioncolor}{rgb}{0.0, 0.3, 0}

\usepackage{cancel}
\usepackage{changes}

\usepackage{todonotes}
\usepackage{mathrsfs}
\usepackage[hyperfootnotes=false]{hyperref}

\usetikzlibrary{positioning}
\usetikzlibrary{shapes,arrows}
\newtheorem{remark}{\bf Remark}
\newtheorem{assumption}{\bf Assumption}
\newtheorem{lemma}{\bf Lemma}
\newtheorem{definition}{\bf Definition}
\newtheorem{problem}{\bf Problem}

\newtheorem{example}{\bf Example}
\newtheorem{corollary}{\bf Corollary}

\newtheorem{thm}{\bf Theorem}

\newcommand{\wxy}[1]{\textcolor{black}{#1}}
\newcommand{\xyw}{\color{black}} 
\newcommand{\xywv}{\color{black}} 
\newcommand{\xywr}{\color{black}} 
\newcommand{\xyr}{\color{black}} 
\newcommand{\hw }{\color{black}}
\newcommand{\hwr }{\color{black}}
\newcommand{\hwz }{\color{black}}
\newcommand{\mg}{\color{black}} 

\newcommand{\col}{\textnormal{col}}

\usepackage{caption}
\usepackage{siunitx}
\allowdisplaybreaks
\usepackage{multicol}
\definecolor{deepblue}{rgb}{0, 0.2, 0.5}

\def\BibTeX{{\rm B\kern-.05em{\sc i\kern-.025em b}\kern-.08em
    T\kern-.1667em\lower.7ex\hbox{E}\kern-.125emX}}
\begin{document}

\onecolumn {\textbf
IEEE Publishing Operations has received your accepted article "Safe Learning Control with Optimality and Stability Guarantees" (Digital Object Identifier or DOI: 10.1109/TAC.2026.3707508).}
\twocolumn
\newpage

\title{
Safe Learning Control with Optimality and Stability Guarantees
}








\author{Xinyang Wang, Hongwei Zhang, Shimin Wang, Wei Xiao, Martin Guay
\thanks{This work was supported by the National Natural Science Foundation of China under projects 62473114 and W2541002, and the Shenzhen Science and Technology Program under project SYSPG20241211173609005. (Corresponding author: Hongwei Zhang)}
\thanks{Xinyang Wang and Hongwei Zhang are with the Shenzhen Key Lab for Advanced Motion Control and Modern Automation Equipments, School of Intelligence Science and Engineering, Harbin Institute of Technology, Shenzhen, Guangdong 518055, China (e-mail: wangxy@stu.hit.edu.cn; hwzhang@hit.edu.cn).  Shimin Wang is with the School of Data Science, Lingnan University, Hong Kong (e-mail: smwang@ln.edu.hk). Wei Xiao is with School of Electrical and Electronic Engineering, Nanyang Technological University, Singapore , with M3S, SMART, and with MIT CSAIL (e-mail: weixy@mit.edu). Martin Guay is with Queen's University, Kingston, ON K7L 3N6, Canada (e-mail: guaym@queensu.ca).}
}

\maketitle

\begin{abstract}
\wxy{
Merely pursuing performance may adversely affect safety, while a conservative policy for safe exploration will degrade the performance. How to guarantee both safety and performance in learning-based control problems is an interesting yet challenging issue.
This paper aims to enhance system performance with a safety guarantee by solving reinforcement learning (RL)-based optimal control problems for nonlinear systems subject to high-relative-degree state constraints and unknown time-varying disturbance/actuator faults.
%
%
A new type of control barrier functions (CBFs), termed high-order reciprocal-based control barrier function, is proposed to handle high-relative-degree constraints, which extends the design of CBFs to enforce robust safety without knowing the disturbance bound.
%
%
The concept of gradient similarity is proposed to quantify the relationship between safety and performance.
Finally, gradient manipulation and adaptive mechanisms are introduced in the model-based safe RL framework to enhance the performance with a safety guarantee.
%
%
Two simulation examples illustrate the efficacy of the proposed algorithms.  
%
}

\end{abstract}


\section{Introduction}
%

%
%

{\xyw In practice, safety, optimality and stability are three important design objectives for control systems, especially for safety-critical and resource-constrained systems \cite{chriat2023optimality, wang2022geometrically}}.
As an effective approach to simultaneously consider stability, optimality and safety,  safe optimal control has attracted increasing attention in the control community in the past few years, which
aims to stabilize dynamical systems while minimizing a user-defined cost function and adhering to certain safety constraints. It has been successfully applied in various scenarios, such as autonomous driving 
\cite{cao2022trustworthy} and trajectory optimization \cite{nakka2021chance}.
%

%
%
 Many approaches, such as model predictive control and safety filter, have been investigated to address safe optimal control problems \cite{wabersich2023data}.
However, these approaches require accurate system information to project the optimal controller into the safe input set, and thus cannot rigorously guarantee safety in the presence of disturbances.
As an effective tool for designing optimal controllers for uncertain systems, reinforcement learning (RL) has been studied for addressing safe optimal control problems by using some technical strategies, such as external knowledge 
\cite{du2023safelight}, reachability analysis 
\cite{selim2022safe}, and worst-case criterion \cite{he2024trustworthy}.

%
%
Despite their success in practical applications, these approaches encounter a serious challenge of computational cost, especially for high-dimensional systems.
%
%
To reduce the computational burden, the concept of control barrier function (CBF) was introduced in \cite{ames2017control} as a safety constraint to certify the forward invariance of a safe set.
Two notions of CBF are commonly utilized, i.e., reciprocal CBF and zeroing CBF. The former goes to infinity when approaching the boundary of the safe set, while the latter vanishes.
{\xyr
In \cite{xiao2021high}, a high-order zeroing CBF is proposed to address high-relative-degree constraints. 
}

Recently, CBF has been integrated into {\xyr model-based} RL techniques to address safe optimal control problems.
The approach from \cite{mahmud2021safety} utilizes the barrier transformation technique to map a constrained system to an unconstrained one, allowing standard RL techniques to be applied directly.
Unfortunately, these techniques are restricted to box safety constraints and cannot adequately address complex safety constraints, such as ellipsoidal and cone safety constraints \cite{murtaza2022safety, long2022safe}, which are frequently encountered in robotic control.
Similar to the penalty method (see Chapter 2.1.1 in \cite{xiao2023safe}), a barrier function-based penalty term is incorporated into the cost function in \cite{marvi2021safe}, taking into account both optimality and safety during the learning process.
However, as argued in \cite{almubarak2021hjb}, the effectiveness of the CBF-based penalty method strongly depends on the proper design of the cost function, and an improper cost function may lead to undesirable learning results.
%
It is also worth noting that approaches \cite{mahmud2021safety, marvi2021safe}, where safety and optimality are tightly coupled, may violate the safety constraints during the learning process. 
%

Different from the approaches \cite{mahmud2021safety, marvi2021safe}, a recent study \cite{max2023safe} introduces a {\xyr model-based safe} RL method that emphasizes learning a performance-driven policy while enforcing safety through a safeguard policy.
In other words, this method simplifies the learning target by decoupling the safety objective from the learning process using the safeguard policy. 
This allows the method to focus solely on optimizing performance, thus reducing the computational cost.
%
%
%
%
This method has been successfully applied in multi-quadrator systems to address safe human-swarm interaction problems \cite{li2024game}.
Although this safeguard-based RL approach can significantly reduce the computational burden, it is limited to specific safety scenarios, such as safety constraints with relative degree one,  leading to a more conservative solution.
{\xyr A high-order zeroing CBF-based safety filter was proposed in \cite{peng2023design} to design an RL approach that addresses high-relative-degree constraints.}
{\xyr However}, the computational cost required for the implementation of the approach proposed in \cite{peng2023design} is significantly more than the approach in \cite{max2023safe} {\xyr since it requires the solution of an optimization problem in real-time. Additionally, high-order zeroing CBF-based safety filter remains sensitive to} disturbance/fault.


%
{\mg The limitations discussed above} motivate an in-depth investigation of safeguard-based RL approaches. 
Specifically, on the basis of \cite{max2023safe}, we extend the {\xyr model-based safe} RL to address optimal control problems subject to high-relative-degree constraints.
Also, given that objectives of stability/optimality and safety may be conflicting, an important objective of this study is the design of a {\mg control system that compromises seemlessly between} safety concerns and performance requirements.
%
%
Furthermore, the robustness of {\xyr model-based} safe RL to {\xyr unbounded} disturbance is explored.
%
First, we introduce a new type of CBFs, termed high-order {\xyr reciprocal-based} CBFs (HO-RCBFs),  
which are capable of {\xyr enforcing robust safety of systems subject to high-relative-degree constraints without any knowledge of disturbances.}
%
The key difference between HO-RCBFs and HO-ZCBFs in \cite{xiao2021high} is that HO-RCBFs {\mg require} large power to push the system trajectory away from the boundary of the safe set, {\mg while} HO-ZCBFs vanish {\mg as it approaches} the boundary of the safe set.
This property {\xyr is used to reject disturbance in HO-RCBFs by reciprocal-based energy instead of using the bound of disturbance.}
%
Second, we extend the safeguard controller from \cite{max2023safe} to a high-order version to handle high-relative-degree safety constraints, and rigorously analyze its robustness to unknown disturbance/fault.
Most importantly, {\mg our approach formally achieves a weighted balance between safety and performance with theoretically guaranteed safety, which has been recognized as a significant challenge in safe RL \cite{gu2024balance}.}

\emph{\textbf{Notations:}} 
Matrix $P>0$ means $P$ is positive definite.
Both the Euclidean norm of a vector and the Frobenius norm of a matrix are denoted by $\lVert \cdot \rVert$, respectively.
The maximum and minimum eigenvalues of a matrix are denoted by $\bar{\sigma}(\cdot)$ and $\underline{\sigma}(\cdot)$, respectively.
For a time-varying bounded signal $d(t)$, $\lVert d \rVert_\infty = \sup_{t \geq0} \lVert d(t) \rVert$.
The inner product of $y_1$ and $y_2$ is denoted by $\langle y_1, y_2\rangle$.
Given any compact set $\mathcal{V} \subset \mathbb{R}^q$ and a continuous mapping $p: \mathcal{V} \rightarrow \mathbb{R}^Q$, $\overline{\lVert p(v) \rVert}_{\mathcal{V}} = \mathop{\rm{sup}}_{v \in \mathcal{V}} \lVert p(v)\rVert$.
{\xyr For a continuously differentiable function $\zeta(x)$, $\nabla \zeta(x) = \partial \zeta(x)/\partial x$, and $\mathcal{L}_f\zeta(x)$ is the Lie derivative of $\zeta$ along $f$.
}


\section{Problem formulation} \label{Sec Preliminaries}
Consider a nonlinear control system
\begin{align} \label{eq:nominal_sys}
    \dot{x} = f\left(x\right)  + g\left(x\right)u
\end{align}
where $x  \in {\xyr \mathcal{X}} \subset \mathbb{R}^n$ and $u \in \mathbb{R}^p$ are the state and  input of the system, respectively; 
%
$f: {\xyr \mathcal{X}} \rightarrow \mathbb{R}^n$ and $g : {\xyr \mathcal{X}} \rightarrow \mathbb{R}^{n \times p}$ are locally Lipschitz; 
{\hwz $\mathcal{X}$ is an admissible set}, and $x(0)=x_0$.
It is assumed that $f(0)=0$, and the system is stabilizable and forward complete. 

\subsection{Unconstrained optimal control}
To seek {\xyr a stabilizing}
control law for \eqref{eq:nominal_sys}, one way is to solve the following {\wxy{optimal control problem \cite{lewis2012optimal}} 
\begin{align}\label{e2.1.1.1}
        \mathop{\rm{inf}}_{u \in \mathbb{R}^p} J(x, u) :=\int_{0}^\infty  \underbrace{x(\tau)^{\top} Q x(\tau) + u(\tau)^{\top} R u(\tau)}_{l\left(x(\tau),u(\tau)\right)} d\tau,
\end{align}
where $J$ is the user-defined cost function and $Q,R > 0$.

To address such an optimal control problem, the optimal value function is defined as 
\begin{align} \label{eq:optimal_value_function}
    V^*(x(t)) = \mathop{\rm{inf}}\limits_{u \in \mathbb{R}^p}\int_{t}^\infty l\left(x(\tau), u(\tau)\right) d\tau
\end{align}
and the Hamiltonian function {\hwz is} 
\begin{align} \label{Ham}
    \hspace{-0.1in}H(x, \nabla V^*(x), u) = \nabla V^*(x)^{\top} (f(x) + g(x) u) + l(x, u),
\end{align}
where $\nabla V^*(x) = \partial V^*(x) / \partial x$ and $V^*(0) =0$.

{\hwz Suppose} a continuously differentiable $V^*(x)$ {\xyr with locally Lipschitz $\nabla V^*(x)$} exists. {\hwz Then} it is the unique positive definite solution of the Hamilton-Jacobi-Bellman (HJB) equation
\begin{align}\label{eq:HJB_without_u}
    \mathop{\rm{inf}}\limits_{u \in \mathbb{R}^p} H(x, \nabla V^*(x), u) = 0.
\end{align}
By applying the optimal condition $\partial H/ \partial u = 0$, the optimal control policy can be derived as
\begin{align} \label{eq:optimal_control_policy}
   k^*(x) = -\frac{1}{2}R^{-1}g(x)^{\top} \nabla V^*(x).
\end{align}
{\xyr 
Using \eqref{eq:optimal_control_policy}, the HJB equation \eqref{eq:HJB_without_u} can be further written as 
\begin{align}\label{eq:HJB}
    \hspace{-0.1in}\nabla V^*(x)^\top f(x)  =-x^\top Qx+ \frac{1}{4}\lVert\sqrt{R^{-1}}g(x)^\top \nabla V^*(x) \rVert^2.
\end{align}}
%
\vspace{-0.5cm}
\subsection{Constrained optimal control}
{\hwz Let} $\mathscr{C}$ be a nonempty safe set with no isolated points and 
\begin{align*}
    \mathscr{C} = \{x \in {\xyr \mathcal{X}}|~ h(x) \geq 0 \},
\end{align*}
where $h : {\xyr \mathcal{X}} \rightarrow \mathbb{R}$ is $m$th-order differentiable. 
Then, a constrained optimal control problem is formulated as follows.
\begin{problem} \label{Pro:nominal_safety}
    {\xyw For system \eqref{eq:nominal_sys}},
    %
   {\hwz design} a control law $u$ such that the optimal control problem \eqref{e2.1.1.1} {\hwz is solved} and the safe set $\mathscr{C}$ is forward invariant, i.e.,
    \begin{align}\label{COCP1}
    &\inf_{u \in \mathbb{R}^p} \; \int_{0}^\infty l(x(\tau), u(\tau)) d\tau  \\
    &\textnormal{s.t.} \quad x(t) \in \mathscr{C}, \; \forall t \geq 0, ~ x_0 \in \mathscr{C}\nonumber.
    \end{align}
\end{problem}

\wxy{For \eqref{COCP1}, the control policy \eqref{eq:optimal_control_policy} is {\mg optimal but generally does not satisfy the constraint}.
%
{\hwz Directly} solving \eqref{COCP1} {\hwz is considered to be} computationally expensive and inefficient \cite{ames2017control}, and thus {\mg {\xyr control barrier functions} (CBFs) are recently applied to provide sub-optimal, but computationally efficient, solutions to constrained optimal control problems.}
%
Before presenting a CBF-based solution, we introduce some basic concepts.}
%


\begin{definition}\label{def:least_relative_degree}
({\xyr Least} relative degree \cite{tan2022high}):
{\xyr Consider} a safe set $\mathscr{C}$.
An {\xyr $m$th-order} differentiable function $h:\mathbb{R}^n \rightarrow \mathbb{R}$ is said to have a least relative degree of {\xyr $m$} with respect to \eqref{eq:nominal_sys}, {\xyr if $ \mathcal{L}_g\mathcal{L}_f^{k} h(x)=0,$ $\forall x \in \mathscr{C}$} and $\forall k\in \{0,1,\dots, m-2\}$.\footnote{The relative degree condition is much weaker than the uniform relative degree condition \cite{xiao2021high} that requires ${\xyr\mathcal{L}_g \mathcal{L}_f^{m-1}h(x)} \neq 0$, $\forall x \in \mathscr{C}$.}
\end{definition}

{\xyr
For the $m$th-order differentiable function $h(x)$ with least relative degree of $m$,} define a series of sets as
\begin{align}\label{eq:HO_set}
    \mathscr{C}_i = \{x\in {\xyr \mathcal{X}}|~\psi_{i-1}(x) \geq 0\}, ~ i\in \{1,\cdots,m\},
\end{align}
where $\psi_i: {\xyr\mathcal{X}}\rightarrow \mathbb{R}$ satisfies
\begin{align}\label{eq:psi_i}
    \psi_i(x) = \dot{\psi}_{i-1}(x) + \alpha_i(\psi_{i-1}(x))
\end{align}
with $\psi_0(x) = h(x)$ and $\alpha_1,~\alpha_2,~\cdots,~\alpha_m$ being sufficiently smooth class $\mathcal{K}$ functions.
{\xyr Let $\tilde{\mathscr{C}}=\cap_{i=1}^m\mathscr{C}_i$.
Then the high-order zeroing CBF (HO-ZCBF) is defined as follows.}
\begin{definition}\cite{xiao2021high}\label{d2}
{\xyr (HO-ZCBF):
Let $\mathscr{C}_i$ and $\psi_{i-1}(x)$ be defined by  \eqref{eq:HO_set} and \eqref{eq:psi_i}, respectively.
An $m$th-order differentiable function $h: \mathcal{X} \rightarrow \mathbb{R}$ is an HO-ZCBF for \eqref{eq:nominal_sys} if $\forall x_0 \in \tilde{\mathscr{C}}$, there exist sufficiently smooth} 
{\xyr class $\mathcal{K}$ functions $\alpha_i,~i\in\{1,\cdots, m\}$ such that
\begin{align*}
    \mathop{\rm{sup}}\limits_{u \in \mathbb{R}^p} \big\{\mathcal{L}_f \psi_{m-1}(x) + \mathcal{L}_g \psi_{m-1}(x)u \big\} \geq -\alpha_m(\psi_{m-1}(x))
\end{align*}
for all $x \in \tilde{\mathscr{C}}$.}
\end{definition}


{\xyr
%
%
As established in \cite{xiao2021high}, the existence of an HO-ZCBF guarantees the existence of a control law $u$ that makes $\psi_{m-1}(x)\geq 0$, thereby rendering the set $\tilde{\mathscr{C}}$ forward invariant for system \eqref{eq:nominal_sys}.
Then one can reformulate Problem \ref{Pro:nominal_safety} as the following optimal control problem
\begin{align*}
    &\inf_{u \in \mathbb{R}^p} \; \int_{0}^\infty l(x(\tau), u(\tau)) d\tau  \\
    &\textnormal{s.t.} ~\mathcal{L}_f \psi_{m-1}(x)+\mathcal{L}_g \psi_{m-1}(x)u \geq -\alpha_m (\psi_{m-1}(x)), ~ x_0 \in \tilde{\mathscr{C}}\nonumber,
\end{align*}
where the optimal solution can be derived as \cite{mahmud2021safety}
\begin{align*}
    &u^*=\arg\inf_{u \in \mathbb{R}^p} \;  H(x,\nabla V^*(x),u)  \\
    &\textnormal{s.t.} ~\mathcal{L}_f \psi_{m-1}(x)+\mathcal{L}_g \psi_{m-1}(x)u \geq -\alpha_m (\psi_{m-1}(x)), ~ x_0 \in \tilde{\mathscr{C}}\nonumber.
\end{align*}

%
%
}

\subsection{{\xyr Constrained optimal control under fault/disturbance}}
{\xyr
Now we consider the following disturbed system}
\begin{align} \label{eq:disturbed_sys}
    \dot{x}= f\left(x\right)  + g\left(x\right) (u + u^f),
\end{align}
where $u^f\in\mathbb{R}^p$ represents the unknown actuator fault and/or matched disturbance\footnote{As an actuator fault, $u^f$ can represent force bias caused by motors in unmanned aerial vehicles \cite{ma2019nonlinear}, actuator bias of distributed generators in microgrids \cite{zhai2022distributed}, etc.
    %
    {\mg In addition,} $u^f$ can represent a matched disturbance, which is a common occurrence in practical applications, such as synchronous motors \cite{yan2018robust} and unmanned vehicles \cite{peng2019constrained}.} {\xyr satisfying the following assumption.}

\begin{assumption} \label{a1} 
    The actuator fault/matched disturbance $u^f$ {\hwz is bounded, i.e., $\lVert u^f (t)\rVert_\infty \leq \eta,$
    where {\hwz $\eta > 0$ is} unknown.}
    %
\end{assumption}

{\xyr
This paper aims to solve the following safe robust control problem.

\begin{problem} \label{Prob-SRC}
Consider the disturbed control system \eqref{eq:disturbed_sys} and the safe set $\mathscr{C}$. 
Suppose Assumption \ref{a1} holds.
Design a control law $u$ such that the disturbed system \eqref{eq:disturbed_sys} is input-to-state stable (ISS) with respect to $u^f$, and simultaneously satisfies the safety constraint $h(x(t))\geq 0$ for all $t\geq0$.    
\end{problem}

In the sequel, we shall show that Problem \ref{Prob-SRC} can be addressed by solving a constrained optimal control problem.  

%
As shown in \cite{sun2024safety}, the optimal value function \eqref{eq:optimal_value_function} is an ISS Lyapunov function of system \eqref{eq:disturbed_sys} with respect to $u^f$. Thus the disturbed system \eqref{eq:disturbed_sys} with $u=k^*(x)$ defined in \eqref{eq:optimal_control_policy} is ISS with respect to $u^f$.
To see this, taking time derivative of $V^*$ along \eqref{eq:disturbed_sys} and using the HJB equation \eqref{eq:HJB} yields
\begin{align}\label{eq:ISS}
    \dot{V}^*(x) = &\ \nabla V^{*}(x)^\top (f(x)+g(x)(k^*(x)+u^f)) \notag\\
    \leq &\ \nabla V^*(x)g(x)u^f-x^\top Q x- \frac{1}{4\bar{\sigma}(R)}\lVert g(x)^\top \nabla V^*(x) \rVert^2\notag\\ 
    \leq& -x^\top Q x + u^{f\top} Ru^f\notag\\
    \leq& -\frac{1}{2}\underline{\sigma}(Q)\lVert x\rVert^2, \forall \lVert x\rVert\geq\sqrt{2\bar{\sigma}(R)/\underline{\sigma}(Q)}\lVert u^f \rVert.
\end{align}
To render the set $\tilde{\mathscr{C}}$ forward invariant for system \eqref{eq:disturbed_sys}, a robust form of an HO-ZCBF constraint is} \cite{jankovic2018robust}
\begin{align} \label{eq:disturbed_CBF_condition}
    &\mathcal{L}_f \psi_{m-1}(x) +   \mathcal{L}_g \psi_{m-1}(x)u \notag\\
    &- \lVert \mathcal{L}_g \psi_{m-1}(x)\rVert ~\lVert u^f \rVert_{\infty} \geq -\alpha_{m}(\psi_{m-1}(x)).
\end{align}
{\xyr Now we are ready to formulate a constrained robust optimal control problem, which addresses Problem \ref{Prob-SRC}.
\begin{problem} \label{Pro:robust_safety}
    Consider the disturbed control system \eqref{eq:disturbed_sys}
    and the safe set $\mathscr{C}$. 
    Suppose Assumption \ref{a1} holds.
    Design a control law $u$ such that 
    \begin{align}\label{COCP2}
    &\inf_{u \in \mathbb{R}^p} \; H(x,\nabla V^*(x),u) \\
    &~~\textnormal{s.t.}~~ \mathcal{L}_f \psi_{m-1}(x)+  \mathcal{L}_g \psi_{m-1}(x)u \notag\\
    &~~~~~~~~~ -\lVert \mathcal{L}_g \psi_{m-1}(x)\rVert ~\lVert u^f \rVert_{\infty} \geq -\alpha_m(\psi_{m-1}(x)), ~x_0 \in \tilde{\mathscr{C}}.\nonumber
    \end{align}
\end{problem}
%


\begin{remark}
%
%
Obviously, the robust HO-ZCBF \eqref{eq:disturbed_CBF_condition} requires the bound of $u^f$, i.e., $\lVert u^f \rVert_{\infty}$; and thus it fails to guarantee safety under Assumption \ref{a1}.
%
%
To avoid using the bound of $u^f$, a robust HO-ZCBF is designed as \cite{ersin2025safety}
\begin{align*}
    &\mathcal{L}_f \psi_{m-1}(x) +   \mathcal{L}_g \psi_{m-1}(x)(u+\hat{u}^f) \\
    &- \lVert \mathcal{L}_g \psi_{m-1}(x)\rVert ~\lVert u^f -\hat{u}^f \rVert_{\infty} \geq -\alpha_m(\psi_{m-1}(x)),
\end{align*}
where $\hat{u}^f$ is the estimation of $u^f$, generated by a disturbance observer. But the bound $\lVert u^f -\hat{u}^f \rVert_{\infty}$ should be known, which further implies that $\lVert \dot{u}^f\rVert_\infty$ is known \cite{ersin2025safety}.
\end{remark}
%

%
%
}

\section{High-order reciprocal-based CBF} \label{Sec-horcbf}
{\xyr
In this section, we propose a new class of robust CBF to enforce safety without any prior knowledge of the disturbance.  
We first consider the case where the least relative degree is $m=1$.
Since $\lVert u^f\rVert_\infty$ is unknown, we construct a continuously differentiable function $b: \operatorname{Int}(\mathscr{C}) \rightarrow \mathbb{R}$ satisfying 
\begin{align}\label{eq:constraint_a}
    \frac{1}{\breve{\gamma}_1\left(h(x)\right)}\leq  b(x)  \leq  \frac{1}{\breve{\gamma}_2\left(h(x)\right)}
\end{align}
for some class $\mathcal{K}$ functions $\breve{\gamma}_1$ and $\breve{\gamma}_2$, where
\begin{align*}
    \operatorname{Int}(\mathscr{C}) = \{x\in \mathcal{X}|~h(x)>0\}.
\end{align*}
The left inequality of \eqref{eq:constraint_a} guarantees that $b(x) \geq \lVert u^f\rVert_\infty$ when $x$ gets close enough to the boundary of set $\mathscr{C}$, i.e.,
$$
b(x) \geq \lVert u^f\rVert_\infty,~ \forall x\in \{x\in\mathcal{X}|~0<h(x)\leq \breve{\gamma}_1^{-1}(\lVert u^f\rVert_\infty)\},
$$
where $\breve{\gamma}_1^{-1}$ is the inverse of $\breve{\gamma}_1$,
and the right inequality of \eqref{eq:constraint_a} controls the growing rate of $b$.
A simple choice of $b$ is $k_b/h$ for arbitrary positive $k_b$.
Condition \eqref{eq:constraint_a} is necessary for the definition of reciprocal CBF (\cite{ames2017control}).
The reciprocal CBF constraint $\dot{b}(x)\leq \breve{\gamma}_3(h(x))$ ($\breve{\gamma}_3$ is a class $\mathcal{K}$ function) for system \eqref{eq:disturbed_sys} is  
\begin{align}\label{eq:rcbf_condition}
    \mathcal{L}_f b(x) + \mathcal{L}_g b(x) (u+u^f) \leq \breve{\gamma}_3(h(x)).
\end{align}
%
As is shown in \cite{ames2017control}, the set $\operatorname{Int}(\mathscr{C})$ is forward invariant if \eqref{eq:rcbf_condition} holds.
Using Young's inequality and $\lVert u^f\rVert_\infty \leq \eta$, a robust reciprocal CBF can be obtained as 
\begin{align*}
    \mathcal{L}_f b(x) + \mathcal{L}_g b(x)u + \lVert \mathcal{L}_g b(x) \rVert^2 + \frac{1}{4} \eta^2 \leq \breve{\gamma}_3(h(x)).
\end{align*}
The following lemma shows how the reciprocal-like function $b(x)$ is used to reject disturbance without knowing its bound.

\begin{lemma}\label{lm:b_rcbf}
    Consider the disturbed system \eqref{eq:disturbed_sys} and the safe set $\mathscr{C}$.
    Suppose Assumption \ref{a1} holds.
    If $x_0 \in \operatorname{Int}(\mathscr{C})$ and there exists a class $\mathcal{K}$ function $\breve{\gamma}_4$ such that
    \begin{align}\label{eq:robust_rcbf}
    \hspace{-0.1in}\mathcal{L}_f b(x) + \mathcal{L}_g b(x)u + \lVert \mathcal{L}_g b(x) \rVert^2+\breve{\gamma}_4(b(x)) \leq \breve{\gamma}_3(h(x))
    \end{align}
    for all $x \in \operatorname{Int}(\mathscr{C})$, then $\operatorname{Int}(\mathscr{C})$ is forward invariant for \eqref{eq:disturbed_sys}.
\end{lemma}
\begin{proof}
    Using Young's inequality, $\dot{b}(x)$ along \eqref{eq:disturbed_sys} is 
    \begin{align*}
        \dot{b}(x) =& \mathcal{L}_f b(x) + \mathcal{L}_g b(x) (u+u^f)\notag\\
        \leq & \mathcal{L}_f b(x) + \mathcal{L}_g b(x) u + \lVert \mathcal{L}_g b(x) \rVert^2 + \frac{1}{4} \eta^2.
    \end{align*}
    By \eqref{eq:constraint_a} and \eqref{eq:robust_rcbf}, $\dot{b}(x)$ can be further put as
    \begin{align}\label{eq:lm_rcbf_b_dot}
        \dot{b}(x)&  \leq -\breve{\gamma}_4(b(x))  +\breve{\gamma}_3(h(x)) + \frac{1}{4} \eta^2\notag\\
        &  \leq -\breve{\gamma}_4(b(x))  +\breve{\gamma}_3 \circ\breve{\gamma}_2^{-1}\bigg(\frac{1}{b(x)}\bigg) + \frac{1}{4} \eta^2,
    \end{align}
    where $\circ$ represents the composition of functions and $\breve{\gamma}_2^{-1}$ denotes the inverse of the class $\mathcal{K}$ function $\breve{\gamma}_2$.
    Let $\tilde{\mathcal{F}}(b) = -\breve{\gamma}_4(b) + \breve{\gamma}_3 \circ\breve{\gamma}_2^{-1}(\frac{1}{b}) + \frac{1}{4}\eta^2$. 
    Then \eqref{eq:lm_rcbf_b_dot} implies that $\dot{b}(x)\leq \tilde{\mathcal{F}}(b)$.
    According to Lemma 4.2 in \cite{khalil2002nonlinear}, $\breve{\gamma}_3 \circ\breve{\gamma}_2^{-1}$ is a class $\mathcal{K}$ function.
    Since class $\mathcal{K}$ function monotonically increases, $\tilde{\mathcal{F}}(b)$ is strictly decreasing on $b\in(0,\infty)$.
    Further, $\lim_{b\rightarrow 0} \tilde{\mathcal{F}}(b) = +\infty$ and $\lim_{b\rightarrow \infty}\tilde{\mathcal{F}}(b) = -\infty$.
    By the Intermediate Value Theorem, let $\xi_b \in(0,\infty)$ be the unique point such that $\tilde{\mathcal{F}}(\xi_b) = 0$.
    Then \( \dot{b}(x) \leq \tilde{\mathcal{F}} (b)\le0\) for all \(b\geq\xi_b\).
    This precludes the state trajectory that approaches the boundary of $\mathscr{C}$.
    To see this, define
       $ \mathcal{C}_b = \big\{x\in\operatorname{Int}(\mathscr{C})|~b(x) \leq \bar{\xi}_b\}   \text{ and } 
        \partial{\mathcal{C}}_b= \big\{x\in\operatorname{Int}(\mathscr{C})|~b(x) = \bar{\xi}_b\big\}$
    for a positive $\bar{\xi}_b$.
    For all $x_0 \in \operatorname{ Int}(\mathscr{C})$, one has $b(x_0) < \infty$ by \eqref{eq:constraint_a}.
    %
    %
    %
    Let $\bar{\xi}_b = \max \{\xi_b, b(x_0)\}$,
    and thus $x_0 \in \mathcal{C}_b$.
    Then $\dot{b}(x)\leq 0$, $\forall x \in \partial{\mathcal{C}}_b$. 
    This shows that $\mathcal{C}_b$ is forward invariant by using Nagumo's Theorem \cite{blanchini2008set} and noting that $\mathcal{C}_b$ is closed relative to $\operatorname{Int}(\mathscr{C})$.
    %
    %
    %
    Note that $\mathcal{C}_b \subset \operatorname{Int}(\mathscr{C})$, which implies that $x(t)\in\operatorname{Int}(\mathscr{C})$ for all $x_0 \in \mathcal{C}_b$.
    %
    %
    Given that the choice of $x_0$ is arbitrary in $\operatorname{Int}(\mathscr{C})$, every initial state yields a trajectory that stays in $\operatorname{Int}(\mathscr{C})$. Therefore, $\operatorname{Int}(\mathscr{C})$ is forward invariant.
\end{proof}

%

Now we are ready to extend the result in Lemma \ref{lm:b_rcbf} to handle arbitrary-relative-degree constraints.
By Nagumo's Theorem \cite{blanchini2008set}, $\psi_{i-1}(x) \geq 0$ if $\dot{\psi}_{i-1}(x) + \alpha_i(\psi_{i-1}(x)) \geq 0$.
For the function $h(x)$ with least relative degree $m$, the recursive construction \eqref{eq:psi_i} remains identical for both  undisturbed system \eqref{eq:nominal_sys} and  disturbed system \eqref{eq:disturbed_sys} except for the last term $\psi_m$.
Specifically, only the last condition $\psi_m(x) = \dot{\psi}_{m-1}(x)+\alpha_m(\psi_{m-1}(x))$ is affected by the matched disturbance/fault $u^f$. Thus we need to use a reciprocal-like term in $\psi_{m}(x)$ to reject unknown disturbance effect.
Note that $\psi_i(x)\geq 0$ can still be used to guarantee that $\psi_{i-1}(x)\geq 0$ for $i=\{1,\dots,m-1\}$.}

Define $${\xyr \bar{\mathscr{C}} = \mathscr{C}_1\cap \cdots\cap\mathscr{C}_{m-1}\cap \operatorname{Int}(\mathscr{C}_m)},$$ where $\operatorname{Int}(\mathscr{C}_m) = \{x\in \mathcal{X}|~\psi_{m-1}(x)>0\}$.
Then the high-order {\xyr reciprocal-based CBF} (HO-RCBF) is defined as follows.

\begin{definition} \label{d4} (HO-RCBF):
    Let $\mathscr{C}_i$ and $\psi_{i-1}(x)$ be defined by  \eqref{eq:HO_set} and \eqref{eq:psi_i} for $i \in \{1, \cdots, m\}$, respectively. 
    A continuously differentiable function $B: \operatorname{Int}(\mathscr{C}_{m}) \rightarrow \mathbb{R}$ is an HO-RCBF for undisturbed system \eqref{eq:nominal_sys} if $\forall x_0 \in \bar{\mathscr{C}}$, there exist class $\mathcal{K}$ functions $\gamma_1$, $\gamma_2$, $\gamma_3$ and $\gamma_4$ such that 
    \begin{align}
    \frac{1}{\gamma_1\left(\psi_{m-1}(x)\right)}&\leq  B(x)  \leq  \frac{1}{\gamma_2\left(\psi_{m-1}(x)\right)},\label{eq:constraint_a1}
    \end{align}
    \begin{align}
     \hspace{-0.15in} \mathop{\rm{inf}}\limits_{u \in \mathbb{R}^p} \big\{\mathcal{L}_f B
   + \mathcal{L}_g &B u {\xyr  +\lVert \mathcal{L}_g B \rVert^2+ \gamma_3(B )}\big\} \leq\gamma_4\left(\psi_{m-1}\right)\label{eq:constraint_b}
    \end{align}
    for all $x \in \bar{\mathscr{C}}$.
    {\xyr If Assumption \ref{a1} holds and $B$ is an HO-RCBF for system \eqref{eq:nominal_sys}, then it is also an HO-RCBF for system \eqref{eq:disturbed_sys}.}
\end{definition}

In the following Lemma, we {\xyr show} that the existence of an HO-RCBF for {\xyr the undisturbed system} \eqref{eq:nominal_sys} implies the existence of a {\xyr control} policy that renders $ \bar{\mathscr{C}}$ forward invariant {\xyr for \eqref{eq:nominal_sys}}.

\begin{lemma} \label{lm1}
    {\xyr  Consider the sets $\mathscr{C}_i$,~$i\in\{1,\cdots,m\}$.
    If $B(x)$ is an HO-RCBF for system \eqref{eq:nominal_sys} and $x_0 \in \bar{\mathscr{C}}$, then any locally Lipschitz controller $u$ satisfying}
    \begin{align}\label{eq:horcbf}
     \mathcal{L}_f B 
    +\mathcal{L}_g B u {\xyr +\lVert \mathcal{L}_g B \rVert^2+\gamma_3(B)} 
    \leq \gamma_4\left(\psi_{m-1}\right)
    \end{align}
    renders the set $\bar{\mathscr{C}}$ forward invariant for system \eqref{eq:nominal_sys}.
\end{lemma}
\begin{proof}
    {\xyr Using \eqref{eq:constraint_a1} and \eqref{eq:constraint_b},  $\dot{B}$ along the system \eqref{eq:nominal_sys} is
    \begin{align} \label{eq:dot_B}
        \dot{B} \leq& -\gamma_3(B) -\lVert \mathcal{L}_gB \rVert^2+ \gamma_4 \circ\gamma_2^{-1}\left(\frac{1}{B}\right),
    \end{align}
    where $\gamma_2^{-1}$ is the inverse of the class $\mathcal{K}$ function $\gamma_2$.

    Consider the following dynamical system 
    \begin{align} \label{eq:dynamical_sys}
        \dot{z}(t)= -\gamma_3(z(t)) + \gamma_4 \circ\gamma_2^{-1}\left(\frac{1}{z(t)}\right),z(0) = B(x_0).
    \end{align}
    %
    According to Lemma 4.2 in \cite{khalil2002nonlinear}, $\gamma_4 \circ\gamma_2^{-1}$ is a class $\mathcal{K}$ function.
    Let $\mathcal{F}(z) = -\gamma_3(z) + \gamma_4\circ\gamma_2^{-1}(1/z)$.
    Then $\lim_{z\rightarrow 0} \mathcal{F}(z) = +\infty$ and $\lim_{z\rightarrow \infty} \mathcal{F}(z) = -\infty$.
    By $x_0 \in \operatorname{Int}(\mathscr{C}_m)$ and \eqref{eq:constraint_a1}, one has $z(0) < \infty$.
    %
    %
    Then following the similar development as in  the proof of Lemma \ref{lm:b_rcbf}, we can show that $z(t)<\infty,~ \forall t\geq 0$.
    %
    %
    Using the Comparison Lemma \cite{khalil2002nonlinear}, the solution of \eqref{eq:dot_B} satisfies
    \begin{align*}
         B(x(t)) \leq z(t) < \infty,~~\forall t\geq 0.
    \end{align*}
    From \eqref{eq:constraint_a1}, $B(x(t))<\infty$ implies $\psi_{m-1}(x(t)) > 0$ for all $t\geq 0$, i.e.,
    $\dot{\psi}_{m-2}(x(t)) + \alpha_{m-1}(\psi_{m-2}(x(t))) > 0.$
    By Nagumo's Theorem \cite{blanchini2008set} and $\psi_{m-2}(x_0)\geq 0$, one has $\psi_{m-2}(x(t))\geq 0$ for all $t\geq 0$, i.e.,
    $\dot{\psi}_{m-3}(x(t)) + \alpha_{m-2}(\psi_{m-3}(x(t))) \geq 0.$
    Again, we have $\psi_{m-3}(x(t))\geq 0, \forall t\geq 0, \forall \psi_{m-3}(x_0)\geq 0$.
    Iteratively, we have $x(t) \in \mathscr{C}_{i}$, $\forall i\in \{1,\cdots,m-1\}$ and $x(t) \in \operatorname{Int}(\mathscr{C}_{m})$ for all $t \geq 0$.
    Therefore, $ \bar{\mathscr{C}}$ is forward invariant.}
\end{proof}

{\xyr
Now consider the disturbed system \eqref{eq:disturbed_sys}.
The following theorem shows that if \eqref{eq:horcbf} holds, then the  set $\bar{\mathscr{C}}$ can be rendered forward invariant for the disturbed system.
\begin{thm} \label{thm:HORCBF}
    Under Assumption \ref{a1}, if $B(x)$ is an HO-RCBF for the undisturbed system \eqref{eq:nominal_sys} and $x_0 \in \bar{\mathscr{C}}$, then any Lipschitz continuous controller $u$ satisfying \eqref{eq:horcbf} renders the set $\bar{\mathscr{C}}$ forward invariant for the disturbed system \eqref{eq:disturbed_sys}.
\end{thm}
\begin{proof}
    Using \eqref{eq:constraint_a1} and \eqref{eq:constraint_b}, $\dot{B}$ along the system \eqref{eq:disturbed_sys} is
    \begin{align*} 
        \dot{B} \leq-\gamma_3(B) -\lVert \mathcal{L}_gB \rVert^2+ \gamma_4 \circ\gamma_2^{-1}\left(\frac{1}{B}\right) +\mathcal{L}_g Bu^f.
    \end{align*}
    Using Young's Inequality, one has $\lVert\mathcal{L}_g B\rVert \lVert  u^f\rVert \leq \lVert \mathcal{L}_g B\rVert^2 +\frac{1}{4} \lVert u^f\rVert^2$, and thus $\lVert \mathcal{L}_g Bu^f \rVert \leq \lVert \mathcal{L}_g B\rVert^2 +\frac{1}{4} \eta^2$.
    Then  
    \begin{align} \label{eq:dot_B_dis_further}
        \dot{B} \leq-\gamma_3(B) + \gamma_4 \circ\gamma_2^{-1}\left(\frac{1}{B}\right) +\frac{1}{4} \eta^2.
    \end{align}
    %
    %
    Since the right-hand side of \eqref{eq:dot_B_dis_further} shares the same monotonicity property with \eqref{eq:lm_rcbf_b_dot}, it then follows the proof of Lemma \ref{lm:b_rcbf} that $B(x(t)) < \infty,~\forall t\geq 0$ for system \eqref{eq:disturbed_sys}.
    Following a similar development as in the proof of Lemma \ref{lm1}, we can show that $B(x(t)) < \infty$ implies the forward invariance of $\bar{\mathscr{C}}$.
\end{proof}
}

{\xywr
\begin{remark}
    The proposed HO-RCBFs can strictly guarantee safety in the presence of disturbances due to the reciprocal barrier structure. 
    However, two limitations still remain.
    First, as the system state approaches the boundary of the set $\mathscr{C}_m$, the HO-RCBFs may generate excessively large control inputs, which can become incompatible with actuator constraints.
    %
    %
    %
    Second, HO-RCBFs are not well defined on the boundary of the set $\mathscr{C}_m$, and measurement noise or small disturbances may be amplified near the boundary due to the reciprocal structure.
    %
\end{remark}
}

By {\xyr Theorem \ref{thm:HORCBF}}, given a {\xyr robust} constrained optimal control problem, one can always convert the safety objective to rendering the set $\bar{\mathscr{C}}$ forward invariant {\mg by constructing an HO-RCBF.} 
{\xyr
Now, we can solve the following optimal control problem}
\begin{align} \label{eq:ocp_final}
    &\inf_{u \in \mathbb{R}^p} \; H(x,\nabla V^*(x),u)  \\
    &\textnormal{s.t.} ~\mathcal{L}_f B+ \mathcal{L}_g B u+ {\xyr\lVert \mathcal{L}_gB \rVert^2  + \gamma_3(B)}\leq \gamma_4\left(\psi_{m-1}\right),~x_0 \in \bar{\mathscr{C}}\notag
\end{align}
{\xyr to get a control law for Problem \ref{Prob-SRC}.}
{\xyr To solve \eqref{eq:ocp_final},} define the Lagrangian function
\begin{align*}
    L(x, \nabla& V^*(x),  u, \lambda ) =\ H(x, \nabla V^*(x), u) \\
    + \lambda& (  \mathcal{L}_f B+ \mathcal{L}_g B u+ {\xyr\lVert \mathcal{L}_gB \rVert^2  + \gamma_3(B)}-\gamma_4\left(\psi_{m-1}\right)) ,
\end{align*}
where $\lambda : \mathbb{R}^n\rightarrow\mathbb{R}$ is the Lagrange multiplier.

Using the Karush-Kuhn-Tucker (KKT) conditions \cite{boyd2004convex}, {\xyr the solution is optimal if the following conditions hold}
\begin{align*}
      {\xyr \frac{\partial L(x, \nabla V^*,  u^*, \lambda^* )}{\partial u}=}\frac{\partial H{\xyr(x, \nabla V^*, u^*)}}{\partial u} + {\xyr \lambda^*} \mathcal{L}_g B(x) ^{\top}  =&\ 0, \\
      {\xyr \lambda^*}\left(  {\xyr \mathcal{L}_f B+ \mathcal{L}_g B u+ \lVert \mathcal{L}_gB \rVert^2  + \gamma_3(B)- \gamma_4\left(\psi_{m-1}\right)} \right)= &\ 0, \\
      {\xyr \mathcal{L}_f B+ \mathcal{L}_g B u+ \lVert \mathcal{L}_gB \rVert^2  + \gamma_3(B)- \gamma_4\left(\psi_{m-1}\right)}\leq &\ 0\\
     {\xyr \lambda^*} \geq &\ 0.
\end{align*}

Solving the above set of equations/inequalities yields the safe optimal control policy \cite{bandyopadhyay2025lagrangian}
\begin{align} \label{KKTpolicy}
    u^*(x) = -{\xyr \lambda^*} R^{-1} \mathcal{L}_g B({\xyr x})^\top +k^*(x),
\end{align}
and the optimal Lagrange multiplier {\xyr $\lambda^*$ is}
$$
\max\bigg(\frac{ \mathcal{L}_f B+ \mathcal{L}_g B k^*{\xyr+ \lVert \mathcal{L}_gB \rVert^2  + \gamma_3(B)}- \gamma_4\left(\psi_{m-1}\right)}{\mathcal{L}_g B R^{-1} \mathcal{L}_g B^{\top}},0\bigg).
$$

{\xyr
Following a similar analysis conducted in Lemma 2 from \cite{bandyopadhyay2025lagrangian}, it follows that \eqref{eq:ocp_final} is feasible if $\mathcal{L}_gB(x)\neq 0$ for all $x \in \bar{\mathscr{C}}$.
If $\mathcal{L}_g\mathcal{L}_f^{m-1}h(x)\neq 0$  for all $x \in \mathscr{C}$, (i.e., $h(x)$ has a uniform relative degree of $m$), one can always construct a $B(x)$ satisfying $\mathcal{L}_gB(x)\neq 0$ for all $x \in \bar{\mathscr{C}}$.
For example, let $B(x) = \frac{1}{\psi_{m-1}(x)}$.
However, the uniform relative degree condition cannot be always satisfied (see Example 1 in \cite{tan2022high}).
When some points $x \in\mathscr{C}$ exist such that $\mathcal{L}_g \psi_{m-1}(x)= 0$, the optimal control problem \eqref{eq:ocp_final} may become infeasible, and consequently, the control law \eqref{KKTpolicy} may not exist.
In the next section, we shall propose a safeguard policy to guarantee safety when $\mathcal{L}_g\mathcal{L}_f^{m-1}h(x)=0$ for some $x \in \mathscr{C}$ (i.e., $h(x)$ has a least relative degree of $m$).
}

\section{High-order safeguard policy design}\label{Sec-adaptive}

\subsection{High-order safeguard policy} \label{section-HO-safeguard}
Considering the control policy \eqref{KKTpolicy}, we replace $\lambda$ with a positive safeguard gain $K_s$ to {\xyr relax the uniform relative degree condition, i.e., we allow the existence of some $x \in \mathscr{C}$ such that $\mathcal{L}_g \mathcal{L}_f^{m-1}h(x) = 0$}.
{\xyr Further, to make $u(0)= 0$, we use}
\begin{align} \label{eq:ho_rcbf_example}
    \mathcal{B}(x) = \frac{1}{2}(B(x) - B(0))^2
\end{align}
to replace $B(x)$.
%
%
This yields the high-order safeguard policy
\begin{align}\label{eq:ho_safeguard}
    u^s = -K_s R^{-1} \mathcal{L}_g \mathcal{B}({\xyr x})^\top.
\end{align}
%
Compared with \eqref{KKTpolicy}, the proposed controller \eqref{eq:ho_safeguard} does not need to compute $\lambda^*$,
which implies that
{\xyr $\mathcal{L}_g B(x)\neq 0$ is not required for all $x \in \bar{\mathscr{C}}$.
Thus, we do not need the uniform relative degree condition to guarantee $\mathcal{L}_g B(x)\neq 0$ for all $x\in \bar{\mathscr{C}}$, thereby relaxing the uniform relative degree condition. 
Also, without computing $\lambda^*$,  computational burden is reduced.} 
%

%
{\xyr Before showing how the safeguard policy \eqref{eq:ho_safeguard} renders the set $\bar{\mathscr{C}}$ forward invariant under the least relative degree condition, we make the following assumption.}
\begin{assumption} \label{a2}
{\xyr For the function $h(x)$ with least relative degree of $m$,}
the following conditions hold:\footnote{The first condition in Assumption \ref{a2} excludes the case where $0 \in {\xyr \partial\mathscr{C}_m}$, as {\xyr HO-RCBFs are not defined} in this scenario.
{\xyr 
The second condition guarantees that the existence of control forces that render the set $\bar{\mathscr{C}}$ forward invariant under the least relative degree condition.}}
\begin{enumerate} 
    \item The origin $0\in\bar{\mathscr{C}}$, and there exists a neighborhood $N_l$ of the origin such that $N_l \cap \partial\mathscr{C}_m\neq \emptyset$.
    \item {\xyr For all $x_0\in \bar{\mathscr{C}}$, there exists a positive constant $\bar{\psi}_{c}$ such that $\lim_{x\rightarrow\partial\mathscr{C}_m}\lVert \mathcal{L}_g \psi_{m-1}(x) \rVert \geq  \bar{\psi}_{c}$.}
\end{enumerate}
\end{assumption}


{\xyr 
The following lemma illustrates that the forward invariance of $\bar{\mathscr{C}}$ can be also verified by $\mathcal{B}(x)$.

\begin{lemma} \label{lm:B_h}
    Consider the undisturbed system \eqref{eq:nominal_sys}.
    Suppose $x_0 \in \bar{\mathscr{C}}$ and Assumptions \ref{a1}-\ref{a2} hold.
    If $\mathcal{B}(x(t)) < \infty$ holds for all $t \geq 0$ along the trajectories of the system \eqref{eq:nominal_sys}, then the set $\bar{\mathscr{C}}$ is forward invariant for the system \eqref{eq:nominal_sys}.
    The same result holds for the disturbed system \eqref{eq:disturbed_sys}.
\end{lemma}
\begin{proof}
    By Assumption \ref{a2},   $B(0)<\infty$. Hence, from \eqref{eq:ho_rcbf_example}, one can see that $\mathcal{B}(x)<\infty$ implies $B(x)<\infty$.
    From \eqref{eq:constraint_a}, $B(x)$ is lower bounded by $\frac{1}{\gamma_1(\psi_{m-1}(x))}$.
    Then $B(x)<\infty$ implies that $\frac{1}{\gamma_1(\psi_{m-1}(x))} < \infty$, and thus $\psi_{m-1}(x) > 0$.
    If $\mathcal{B}(x(t))<\infty$ holds for all $t\geq 0$ along system \eqref{eq:nominal_sys} (system \eqref{eq:disturbed_sys}), then $\psi_{m-1}(x(t))>0$, $\forall t\geq 0$.
    Following the proof in Theorem \ref{thm:HORCBF}, the set $\bar{\mathscr{C}}$ is forward invariant if $\psi_{m-1}(x(t))>0$, $\forall t\geq 0$.
    Thus, one can conclude that for system \eqref{eq:nominal_sys} (system \eqref{eq:disturbed_sys}), $\bar{\mathscr{C}}$ is forward invariant if $\mathcal{B}(x(t)) < \infty$, $\forall t\geq 0$.
\end{proof}
}

{\xyr Suppose the admissible set $\mathcal{X}$ is compact in the sequel.}
The following result illustrates that {\xyr the safeguard policy \eqref{eq:ho_safeguard} can still guarantee $\bar{\mathscr{C}}$ forward invariant for system \eqref{eq:nominal_sys} (system \eqref{eq:disturbed_sys}) under the least relative degree condition.}
%

\begin{thm} \label{thm:ho_safeguard_safety}
    {\xyr 
    Consider the function $h(x)$ with least relative degree of $m$ and the sets $\mathscr{C}_i,~i\in \{1,\cdots,m\}$.
    Let $\mathcal{H}: \mathbb{R} \rightarrow \mathbb{R}$ be a continuously differentiable function satisfying
    \begin{align} \label{eq:ho_safeguard_condition}
        \frac{1}{\beta_1(s)}\leq\mathcal{H}(s) \leq\frac{1}{\beta_2(s)}~~\textnormal{and}~~\lim_{s\rightarrow0} \left\lvert\frac{\partial \mathcal{H}(s)}{\partial s}\right\rvert= \infty
    \end{align}
    for some class $\mathcal{K}$ functions $\beta_1,\beta_2$.
    If Assumptions \ref{a1}-\ref{a2} hold, $x_0 \in \bar{\mathscr{C}}$ and $B(x) = \mathcal{H}(\psi_{m-1}(x))$, then the control policy $u=u^s(x)$ as defined in \eqref{eq:ho_safeguard} renders the set $\bar{\mathscr{C}}$ forward invariant for both undisturbed system \eqref{eq:nominal_sys} and disturbed system \eqref{eq:disturbed_sys}.}    
\end{thm}
{\xyr
\begin{proof}
    For the undisturbed system \eqref{eq:nominal_sys}, let $f_{cl}(x) = f(x)$.
    By the chain rule, one can express $\dot{\mathcal{B}}(x)$ along \eqref{eq:nominal_sys} as
    \begin{align*}
        \dot{\mathcal{B}}(x) = & \frac{\partial \mathcal{B}(x)}{\partial \psi_{m-1}(x)}(\mathcal{L}_{f_{cl}} \psi_{m-1}(x) + \mathcal{L}_g \psi_{m-1}(x) u^s(x) )  \\
        \leq &\left\lvert\frac{\partial \mathcal{B}(x)}{\partial \psi_{m-1}(x)}\right\rvert^2 \big(- K_s  \mathcal{L}_g \psi_{m-1}(x) R^{-1} \mathcal{L}_g\psi_{m-1}(x)^{\top}\\
        & +\left\lvert\frac{\partial \mathcal{B}(x)}{\partial \psi_{m-1}(x)}\right\rvert^{-1}\overline{\lVert\mathcal{L}_{f_{cl}}\psi_{m-1}(x)\rVert}_{\mathcal{X}}\bigg).
    \end{align*}
    %
    Since $f(x)$ is locally Lipschitz and $\psi_{m-1}(x)$ is continuously differentiable, both $\overline{\lVert f_{cl}(x)\rVert}_{\mathcal{X}}$ and $\overline{\lVert\nabla \psi_{m-1}(x)\rVert}_{\mathcal{X}}$ are bounded.
    By $\overline{\lVert \mathcal{L}_{f_{cl}} \psi_{m-1}(x)\rVert}_{\mathcal{X}} \leq \overline{\lVert f_{cl}(x)\rVert}_{\mathcal{X}}\overline{\lVert\nabla \psi_{m-1}(x)\rVert}_{\mathcal{X}}$, $\overline{\lVert \mathcal{L}_{f_{cl}} \psi_{m-1}(x)\rVert}_{\mathcal{X}}$ is bounded.
    From \eqref{eq:ho_rcbf_example}, we have $$\frac{\partial \mathcal{B}(x)}{\partial \psi_{m-1}(x)}=(\mathcal{H}(\psi_{m-1}(x)) - \mathcal{H}(\psi_{m-1}(0))) \frac{\partial \mathcal{H}(\psi_{m-1}(x))}{\partial \psi_{m-1}(x)}.$$
    By Assumption \ref{a2}, one has $\psi_{m-1}(0)>0$.
    The condition \eqref{eq:ho_safeguard_condition} implies $\mathcal{H}(\psi_{m-1}(0))<\infty$ and $\lim_{x\rightarrow\partial\mathscr{C}_m}\mathcal{H}(\psi_{m-1}(x))=\infty$. 
    Then one has $\lim_{x\rightarrow\partial \mathscr{C}_m} \left\lvert \frac{\partial \mathcal{B}(x)}{\partial \psi_{m-1}(x)}\right\rvert =\infty$.
    Since $\overline{\|\mathcal{L}_{f_{cl}}\psi_{m-1}(x)\|}_{\mathcal{X}}$ is bounded, it follows that
    \begin{align*}
        \lim_{x\rightarrow\partial \mathscr{C}_m}\left\lvert \frac{\partial \mathcal{B}(x)}{\partial \psi_{m-1}(x)}\right\rvert^{-1}\overline{\lVert \mathcal{L}_{f_{cl}} \psi_{m-1}(x)\rVert}_{\mathcal{X}} = 0.
    \end{align*}
    By $\mathcal{L}_g \psi_{m-1}(x) R^{-1} \mathcal{L}_g\psi_{m-1}(x)^{\top} \geq \bar{\sigma}(R)^{-1}\lVert \mathcal{L}_g\psi_{m-1}(x) \rVert^2$ and Assumption \ref{a2}, one has
    \begin{align*}
    \lim_{x\rightarrow\partial \mathscr{C}_m}\bigg(\left\lvert\frac{\partial \mathcal{B}(x)}{\partial \psi_{m-1}(x)}\right\rvert^{-1}&\overline{\lVert \mathcal{L}_{f_{cl}}\psi_{m-1}(x)\rVert}_{\mathcal{X}}\\
    - K_s\mathcal{L}_g \psi_{m-1}(x) R^{-1} \mathcal{L}_g\psi_{m-1}&(x)^{\top} \bigg) \leq -K_s\bar{\sigma}(R)^{-1}\bar{
    \psi}_{c}^2.
    \end{align*}
    Hence, one has $\lim_{x\rightarrow\partial \mathscr{C}_m} \dot{\mathcal{B}}(x) < 0$.
    According to the proof of Theorem 1 in \cite{max2023safe}, this precludes the existence of trajectories that enter $\partial\mathscr{C}_m$, and thus $\mathcal{B}(x(t))< \infty$ for all $t\geq 0$.
    It then follows from Lemma \ref{lm:B_h} that the set $\bar{\mathscr{C}}$ is forward invariant for the undisturbed system \eqref{eq:nominal_sys}.

    For the disturbed system \eqref{eq:disturbed_sys}, redefine $f_{cl}(x)=f(x) + g(x)u^f$.
    Since $g(x)$ is locally Lipschitz, $\overline{\lVert g(x) \rVert}_{\mathcal{X}}$ is bounded.
    Then the fact that  $\overline{\lVert\nabla \psi_{m-1}(x)\rVert}_{\mathcal{X}}$ is bounded and $\overline{\lVert \mathcal{L}_{g} \psi_{m-1}(x)\rVert}_{\mathcal{X}} \leq \overline{\lVert g(x)\rVert}_{\mathcal{X}}\overline{\lVert\nabla \psi_{m-1}(x)\rVert}_{\mathcal{X}}$ illustrate the boundedness of $\overline{\lVert\mathcal{L}_{g}\psi_{m-1}(x)\rVert}_{\mathcal{X}}$.
    Recalling $\lVert u^f \rVert_{\infty}\leq \eta$ and $\overline{\lVert\mathcal{L}_{f_{cl}}\psi_{m-1}(x)\rVert}_{\mathcal{X}} \leq \overline{\lVert\mathcal{L}_{f}\psi_{m-1}(x)\rVert}_{\mathcal{X}} + \overline{\lVert\mathcal{L}_{g}\psi_{m-1}(x)\rVert}_{\mathcal{X}}\lVert u^f \rVert_{\infty}$, we can show that $\overline{\lVert\mathcal{L}_{f_{cl}}\psi_{m-1}(x)\rVert}_{\mathcal{X}}$ is bounded.
    It then follows from the proof for the undisturbed system \eqref{eq:nominal_sys} that the set $\bar{\mathscr{C}}$ is forward invariant for the disturbed system \eqref{eq:disturbed_sys} under \eqref{eq:ho_safeguard}.
\end{proof}
}

{\xyr 
Two illustrative examples of $\mathcal{H}$ are $$\mathcal{H}(s) = \frac{1}{s}~\textnormal{and}~\mathcal{H}(s) = -\log\left(\frac{s}{1+s}\right),$$ whose derivatives are $\lvert \frac{\partial \mathcal{H}(s)}{\partial s}\rvert  =\frac{1}{s^2}$ and $\lvert \frac{\partial \mathcal{H}(s)}{\partial s}\rvert  =\frac{1}{s(s+1)}$, respectively.
Both functions clearly satisfy the condition in \eqref{eq:ho_safeguard_condition}.
}
In the following corollary, we show how the safeguard policy \eqref{eq:ho_safeguard} guarantees safety of both systems \eqref{eq:nominal_sys} and \eqref{eq:disturbed_sys} under {\xyr an existing} control policy.

\begin{corollary} \label{corol:us_kx}
        %
        {\xyr
        Let $\mathcal{H}: \mathbb{R} \rightarrow \mathbb{R}$ be a continuously differentiable function satisfying \eqref{eq:ho_safeguard_condition} and $k(x)$ be a locally Lipschitz control policy on $x\in \mathcal{X}$.
        If Assumptions \ref{a1}-\ref{a2} hold, $x_0 \in \bar{\mathscr{C}}$ and $B(x) = \mathcal{H}(\psi_{m-1}(x))$, then the control policy
        \begin{align*}
            u = k(x) + u^s(x),
        \end{align*}
        with $u^s$ defined in \eqref{eq:ho_safeguard},
        renders the set $\bar{\mathscr{C}}$ forward invariant for both undisturbed system \eqref{eq:nominal_sys} and disturbed system \eqref{eq:disturbed_sys}.
        }
\end{corollary}
{\xyr
\begin{proof}
    For the undisturbed system \eqref{eq:nominal_sys}, redefine $f_{cl}(x) = f(x)+g(x)k(x)$ and then the proof can be finished by showing the boundness of $\overline{\lVert f_{cl}(x) \rVert}_{\mathcal{X}}$.
    Since $k(x)$ is locally Lipschitz, $\lVert k(x)\rVert$ is bounded for all $x \in \mathcal{X}$, and further $\overline{\lVert f_{cl}(x)\rVert}_{\mathcal{X}}$ is bounded.
    Following the proof of Theorem \ref{thm:ho_safeguard_safety}, one can show that $\bar{\mathscr{C}}$ is forward invariant for the undisturbed system \eqref{eq:nominal_sys}.
    For the disturbed system \eqref{eq:disturbed_sys}, we redefine $f_{cl}(x) = f(x)+g(x)(u^f+k(x))$.
    Recalling $\lVert u^f \rVert_{\infty}\leq \eta$, $\lVert f_{cl}(x)\rVert$ is bounded on $\mathcal{X}$.
    It then follows from the proof of Theorem \ref{thm:ho_safeguard_safety} that the set $\bar{\mathscr{C}}$ is  forward invariant for the disturbed system \eqref{eq:disturbed_sys}.
\end{proof}
}


\subsection{Balance between safety and performance}

{\xyr According to Corollary \ref{corol:us_kx}}, the control law $$u=k^*(x)+u^s(x)$$ {\xyr can render the set $\bar{\mathscr{C}}$ forward invariant}. However, its performance ceases to be optimal, with the Hamiltonian being 
\begin{align}  \label{origin-H}
    H(x, \nabla V^*, k^*(x) + u^s(x)) = &\ \frac{d V^*}{dt} + l(x, k^*(x) + u^s(x))\notag \\
    = &\ K_s^2 \lVert \sqrt{R^{-1}}\mathcal{L}_g \mathcal{B} \rVert^2.
\end{align}
The above statement implies that performance and safety are often conflicting objectives. 
And the {\mg extent} of this conflict can be described by the following \textit{gradient similarity} measure }
\begin{align} \label{cos}
    \rho  = \cos(\theta(x))  = \frac{\langle \sqrt{R^{-1}}g^{\top} \nabla V^*, \sqrt{R^{-1}}\mathcal{L}_g \mathcal{B} \rangle}{\lVert \sqrt{R^{-1}}g^{\top} \nabla V^* \rVert \lVert \sqrt{R^{-1}}\mathcal{L}_g \mathcal{B} \rVert},
\end{align}
where $\theta\in [0,\pi]$ is the angle between $u^s$ and $k^*$, and $\rho \in [-1, 1]$ illustrates the degree of conflict between safety and performance.
A larger value of $\rho$ indicates less conflict between safety and performance, while a smaller one implies that safety significantly conflicts with performance.
%
%
The time derivative of $V^*$, i.e., 
\begin{align*}
    \dot{V}^* & = \nabla V^{*T} (f(x) + g(x)(k^*(x) + u^s(x))) \\
    & = -l(x,k^*(x)) - \rho\lVert \sqrt{R^{-1}}g^{\top} \nabla V^* \rVert \lVert \sqrt{R^{-1}}\mathcal{L}_g \mathcal{B} \rVert,
\end{align*}
{\hw and the Hamiltonian \eqref{origin-H}} {\mg indicate that the safeguard policy $u^s$ 
leads to a deterioration of the value function when $\rho < 0$ (i.e, $\theta \geq \frac{\pi}{2}$).}
%
Hence, we need to increase $\rho$ when $\theta \geq \frac{\pi}{2}$ to reduce the extent of conflict between $k^*$ and $u^s$. 
By gradient projection, we can divide the safeguard gradient into two parts, i.e.,
$$\sqrt{R^{-1}}{\mathcal{L}}_g \mathcal{B} = \sqrt{R^{-1}}\mathcal{L}_g^- \mathcal{B} + \sqrt{R^{-1}}\mathcal{L}_g^+ \mathcal{B},$$ with
\begin{align*}
    \sqrt{R^{-1}}\mathcal{L}_g^- \mathcal{B} & =\rho \frac{\lVert \sqrt{R^{-1}}\mathcal{L}_g \mathcal{B} \rVert}{\lVert \sqrt{R^{-1}}g^{\top} \nabla V^* \rVert} \sqrt{R^{-1}} g^{\top} \nabla V^* \notag\\
    \sqrt{R^{-1}}\mathcal{L}^+_g \mathcal{B} & =  \sqrt{R^{-1}}\mathcal{L}_g \mathcal{B} - \sqrt{R^{-1}}\mathcal{L}_g^- \mathcal{B} 
\end{align*}
where $\mathcal{L}_g^- \mathcal{B}$ is parallel {\mg to} the direction of the optimal control policy $k^*$ and thus affects the magnitude of $k^*$; and $\mathcal{L}_g^+ \mathcal{B}$ is perpendicular {\mg to} $k^*$, which can be regarded as a centripetal force to change the direction of $k^*$ (see Fig. \ref{fig-gradienthm:ho_safeguard_safety}).
To minimize \eqref{origin-H} while simultaneously taking safety into account, we introduce a parameter $\mu \in [0,1]$ to {\xyr modify} $u^s$  as
\begin{align} \label{eq:adaptive_safeguard_policy}
    u^s(x) = -K_s R^{-1} \tilde{\mathcal{L}}_g \mathcal{B}(x)
\end{align}
where $\sqrt{R^{-1}}\tilde{\mathcal{L}}_g \mathcal{B} =  (1 - \mu ) \sqrt{R^{-1}}\mathcal{L}_g^- \mathcal{B} + \sqrt{R^{-1}}\mathcal{L}_g^+ \mathcal{B}$ (see Fig. \ref{fig-gradient2}). 

%


\begin{figure}[htb]
    \centering
    \begin{subfigure}{0.49\linewidth} 
        \centering
        \includegraphics[width=1.00\textwidth,trim=0 220 0 180]{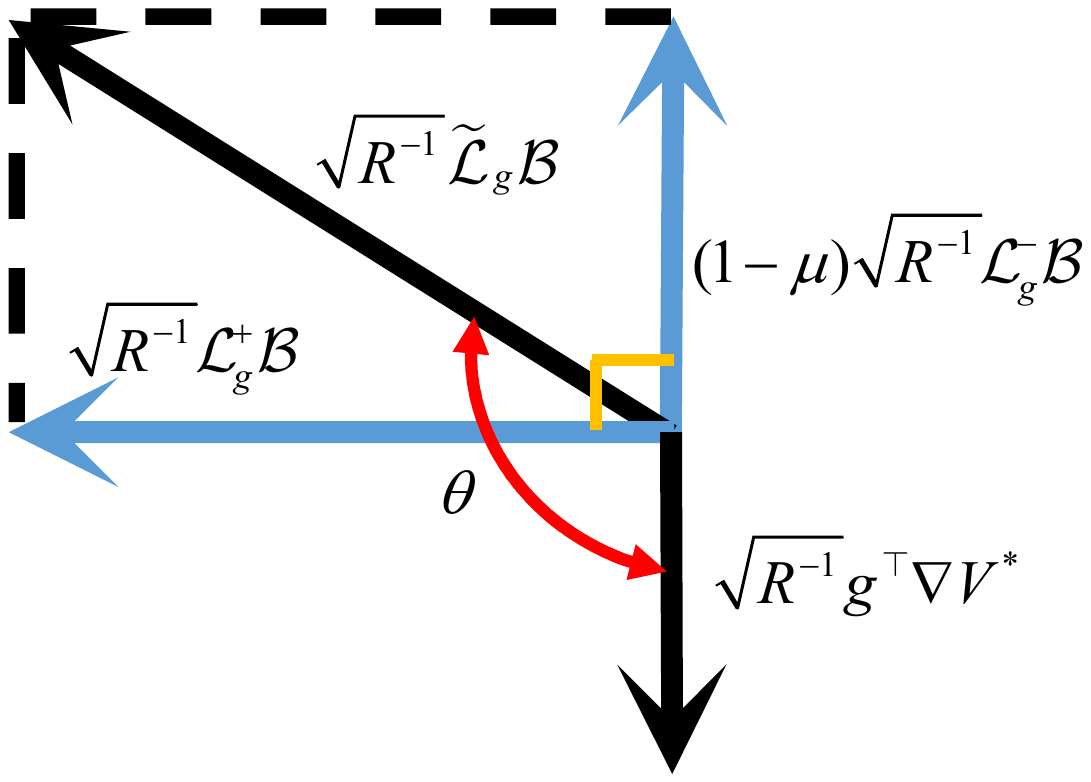} 
        \caption{$\mu = 0$}
        \label{fig-gradienthm:ho_safeguard_safety}
    \end{subfigure}
    \hfill 
    \begin{subfigure}{0.49\linewidth} 
        \centering
        \includegraphics[width=1.0\textwidth,trim=0 220 0 300]{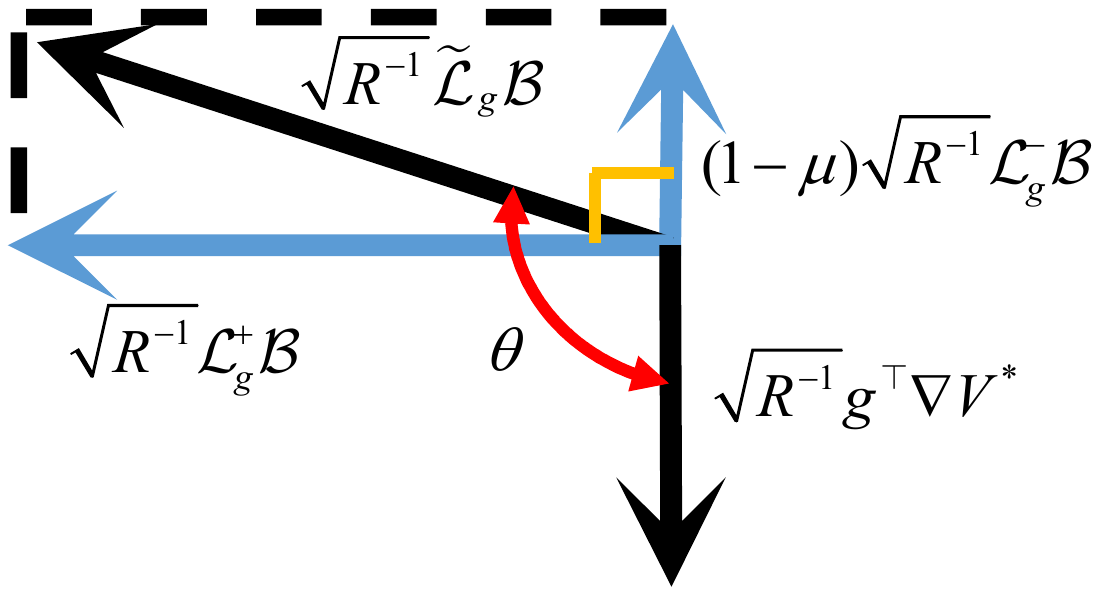} 
        \caption{$\mu = 0.5$}
        \label{fig-gradient2}
    \end{subfigure}
    \caption{Gradient manipulation on $\sqrt{R^{-1}}\tilde{\mathcal{L}}_g \mathcal{B}$.}
    \label{gradient manipulation}
\end{figure}

As shown in Fig.\ref{gradient manipulation}, the gradient similarity $\rho$ increases monotonically with $\mu$ when $\theta \geq \frac{\pi}{2}$, thus enhancing the performance. This is rigorously analyzed in the following lemma. 

%
\begin{lemma} \label{lemma-performance}
    Let $V^*$ be a continuously differentiable positive definite solution of \eqref{eq:HJB_without_u}.
    Then the following inequality holds 
    \begin{align} \label{inequality performance}
        \int_{0}^\infty l(x(\tau), k^*(x(\tau))-K_s R^{-1} \mathcal{L}_g \mathcal{B}(x(\tau))) d\tau  \notag\\
        \geq \int_{0}^\infty l(x(\tau), k^*(x(\tau))-K_s R^{-1} \tilde{\mathcal{L}}_g \mathcal{B}(x(\tau))) d\tau.
    \end{align}
\end{lemma}
\begin{proof}
    Substituting $u(x) =k^*(x) -K_s R^{-1} \tilde{\mathcal{L}}_g \mathcal{B}(x)$ into \eqref{Ham} {\xyr and using the HJB equation as in \eqref{eq:HJB}}, one has
    \begin{align}\label{eq:modified_Hamiltonian}
       &  H(x, \nabla V^*, k^*(x) -K_s R^{-1} \tilde{\mathcal{L}}_g \mathcal{B}(x))  \notag \\
    = &\ \nabla V^{*T} (f(x) + g(x) (k^*(x) -K_s R^{-1} \tilde{\mathcal{L}}_g \mathcal{B}(x)) ) + x^{\top} Q x  \notag \\
    & + (k^*(x) -K_s R^{-1} \tilde{\mathcal{L}}_g \mathcal{B}(x))^{\top} R(k^*(x) -K_s R^{-1} \tilde{\mathcal{L}}_g \mathcal{B}(x))\notag \\
    = &\ K_s^2  (1 - 2 \rho^2 \mu + \rho^2 \mu ^2) \lVert \sqrt{R^{-1}}\mathcal{L}_g \mathcal{B}(x) \rVert^2.
    \end{align}
    Recalling \eqref{origin-H}, {\xyr $\lvert \rho \rvert\leq1$} and the fact that $2 \mu - \mu ^2 \geq 0$ (since $\mu \in [0, 1)$), we have 
     \begin{align*}
     &H(x, \nabla V^*, k^*(x) -K_s R^{-1} \tilde{\mathcal{L}}_g \mathcal{B}(x))\\
     \leq&\ H(x, \nabla V^*, k^*(x) -K_s R^{-1} \mathcal{L}_g \mathcal{B}(x)).\end{align*}
    Then the inequality \eqref{inequality performance} holds by following the {\xyr similar development as in the} proof of Theorem 10.1-2 in \cite{lewis2012optimal}.
\end{proof}

\begin{remark}
{\xyr By \eqref{origin-H} and \eqref{eq:modified_Hamiltonian}}, the performance can be further improved by decreasing $K_s$.
{\xyr Although the safety guarantee under arbitrary $K_s > 0$ has been shown in Theorem \ref{thm:ho_safeguard_safety} for continuous-time system}, an unreasonably small $K_s$ can make the system unsafe {\xyr in practical applications, when the controller is usually implemented in discrete-time using a zero-order hold (ZOH).
As a result, the system under $u^s$ with a small $K_s$ can easily violate safety constraints during the sampling period if the control frequency at which the system measurements are fed back to recalculate the control input is not large enough.}
\end{remark}

A natural question is how small $K_s$ should be {\xyr in practical situations}. 
Before answering this question, {\mg we analyze how safety is affected by changes in $K_s$.} 
%
Suppose $\psi_{m-1}(0) = 0$ and then select an HO-RCBF satisfying $\lvert \partial \mathcal{B}/\partial\psi_{m-1} \rvert =  \alpha_B (\mathcal{B})$, (e.g., $\mathcal{B} = \frac{1}{\psi_{m-1}}$ and $\lvert \partial \mathcal{B}/\partial\psi_{m-1} \rvert= \mathcal{B}^2 $), where $\alpha_B$ is a class $\mathcal{K}$ function.
{\xyr 
Assume that $\lVert \mathcal{L}_g \psi_{m-1}(x)\rVert^2\geq \tilde{\psi}_c$ for a positive $\tilde{\psi}_c$.
Then the time derivative of $\mathcal{B}$ along \eqref{eq:nominal_sys} is}
\begin{align*}
    \dot{\mathcal{B}} \leq \alpha_B (\mathcal{B}) ( \overline{\lVert \mathcal{L}_f \psi_{m-1} \rVert}_{\mathcal{X}} - K_s\tilde{\psi}_c \bar{\sigma}(R)^{-1}  \alpha_B (\mathcal{B})).
\end{align*}
Then we can obtain the upper bound of $\mathcal{B}$ as 
$${\xyr \mathcal{B}(x(t)) \leq \max \left( \mathcal{B}(0),~\alpha_B^{-1} \left(K_s^{-1}\tilde{\psi}_c^{-1}\bar{\sigma}(R) \overline{\lVert \mathcal{L}_f \psi_{m-1} \rVert}_{\mathcal{X}}\right)\right),}$$
which implies that a small $K_s$ not only enhances performance but also allows a large value of $\mathcal{B}$.
The influence of $K_s$ on safety guarantee is demonstrated in the following example.

%
%

%
\begin{example} \label{example1}
Consider a mobile robot described by $\dot{p} = v$ and $\dot{v} = u$, where $p$ is the position, $v$ is the velocity, $u$ is the controller. The controller is designed as $u = k(p, p_d, v) + u^s(p, v)$,  where  $p_d$ is a desired position trajectory, $k = - v + 100(p_d -p)$ steers $p$ to $p_d$,  and $u^s$ is the safeguard policy.
{\xyr Consider a} velocity constraint described by ${\xyr h(t)= v_{max} - v(t) \geq 0}$, where $v_{max}$ is the maximum of speed.
We set the control frequency at $f_c = 100$ Hz, and the sampling time at $0.01\textnormal{s}$.
We then define the HO-RCBF as $B = \frac{1}{h(v)}$.
The velocity and acceleration trajectories of the robot under different $K_s$ are given in Fig. \ref{velocity under Ks}.
\begin{figure}[htbp] 
    \centering
    \includegraphics[width=0.5\textwidth, clip]{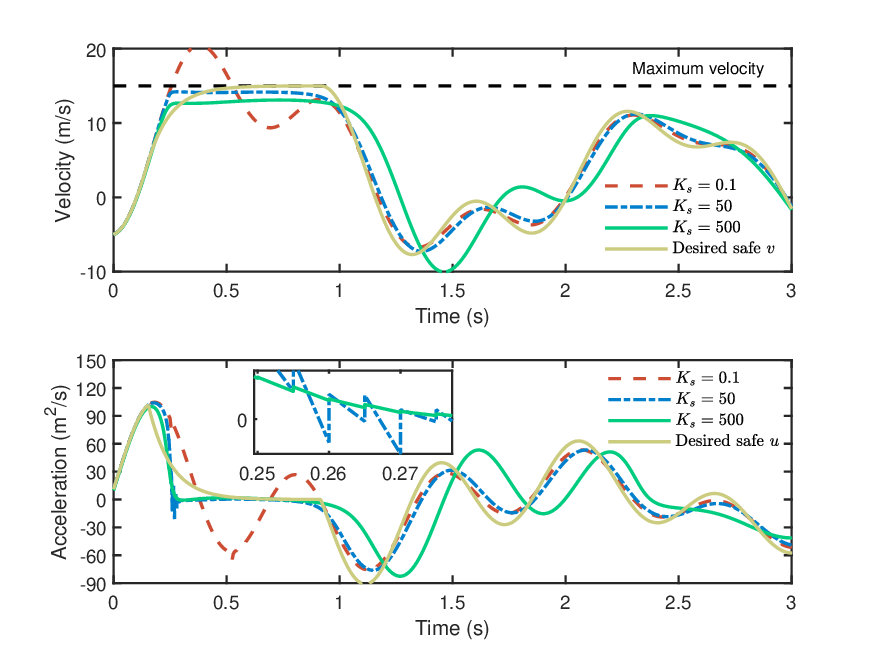}   
    \caption{Velocity and acceleration trajectory of the mobile robot under $K_s = 0.1,~50~\textnormal{and}~ 500$ when $v_{max} = 15 \textnormal{m/s}$. 
    %
    } \label{velocity under Ks}
\end{figure}
When $K_s = 0.1$, the {\xyr velocity} constraint is violated due to the limitation of control frequency.
Suppose $K_s$ is small enough such that $$h(t_k) = h_\epsilon\ \ \ \ \textnormal{and} \ \ \ \ \lVert u^s(t_k) \rVert \leq k(t_k) + \epsilon_u$$ for some $x$ at $t_k$, and $h_\epsilon, \epsilon_u > 0$.
Then one has 
\begin{align*}\dot{h}(t) \leq -k(t_k) + \lVert u^s(t_k) \rVert \leq -\epsilon_u < 0 , & & \forall t \in [t_k, t_{k+1} )\end{align*}
which indicates $h(t) <  0$ if $f_c < \frac{\epsilon_u}{h_\epsilon}$.
Moreover, when $K_s = 50$, the acceleration oscillates because the performance gradient and safety gradient switch rapidly when $v$ approaches $v_{max}$ (see \cite{gu2024balance} for more detailed discussion).
Although a larger $K_s$ can address this issue, it also decreases the tracking performance by imposing a more conservative value of $v$ (see, {\it e.g.}, the case when $K_s=500$).
\end{example}

Hence, a key question is how to design a suitable $K_s$ to guarantee the safety of the system while maintaining its performance. This motivates the following adaptive design approach. 
{\mg In practice, a large $K_{s,0} > 0$ is initially needed to guarantee safety. Then $K_s$ can be updated adaptively to gain better system performance. 
To balance safety and performance, $K_s$ should be decreased to enhance performance when $\mathcal{B}$ is within a safety range and increased to guarantee safety when the control frequency} 
is not large enough or the optimal control is large with a conflicting relationship $\rho<0$.
Here, we define an auxiliary dynamical system for $K_s$ of the form
\begin{align} \label{auxiliary dynamics}
    \dot{K}_s = F(x, K_s) + G(x, K_s) v
\end{align}
where $F : \mathbb{R}^{p+1} \rightarrow \mathbb{R}$ and $G : \mathbb{R}^{p+1} \rightarrow \mathbb{R}^{1 \times q}$ are {\mg chosen such that} $\lVert \mathcal{L}_G K_s\rVert \neq 0 \; \forall x \in \mathbb{R}^n, K_s \in \mathbb{R}$, and $v\in \mathbb{R}^q$ is the control input.
As illustrated in Theorem \ref{thm:ho_safeguard_safety}, a positive $K_s$ can guarantee the set $\bar{\mathscr{C}}$ forward invariant.
Then we can regard $K_s$ as a CBF, and define a constrained input set $\mathcal{U}_k$ for $v$:
\begin{align*}
   \mathcal{U}_k=\{v\in\mathbb{R}^q:~\mathcal{L}_F K_s + \mathcal{L}_G K_s v + \beta_K (K_s)\geq0\},
\end{align*}
where $\beta_K$ is a class $\mathcal{K}$ function {\xyr to allow the decrease of $K_s$.}
For simplicity, let $F = 0$ and $G = I$. Then we can let $$v =-\beta_K (K_s) + M(x, K_s),$$ where $M$ is a user-defined positive semi-definite function.

\subsection{Adaptive high-order safeguard policy}

Combining the gradient manipulation approach and the adaptive mechanism of $K_s$, we propose an adaptive high-order safeguard policy
\begin{align}
   u^s & = -K_s R^{-1} \left((1 - \mu ) \mathcal{L}_g^- \mathcal{B} + \mathcal{L}_g^+ \mathcal{B}\right)
    \label{e3.2.1a} \\
   \dot{K}_s & = \operatorname{Proj}_{\mathcal{S}} \{- \beta_K (K_s) + M(x, K_s) \} \label{e3.2.1b} 
\end{align}
with $K_s(0) = K_{s,0}$ and 
\begin{align*}
    \beta_K (K_s) = Y K_s^2,& &  M(x, K_s) = \gamma  e^{-h}l(x,k^*(x)),
\end{align*}
{\mg where $\gamma > 0$ is a weight used to increase $K_s$, ensuring safety; and $Y > 0$ is a bounded value used to decrease $K_s$ if performance improvements are required. Thus the design parameters $Y$ and $\gamma$ can be used to balance the system's safety and performance.}
%
The operator $\operatorname{Proj}_{\mathcal{S}} \{\cdot \}$ guarantees that $K_s$ stays in a compact set $\mathcal{S}$. 

In the following lemma, we show that the adaptive safeguard policy \eqref{e3.2.1a} and $\eqref{e3.2.1b}$ can render $\bar{\mathscr{C}}$ forward invariant.
\begin{lemma} \label{lm:adaptive_safeguard_safety}
    {\xyr Consider the function $h(x)$ with  least relative degree of $m$ and the sets $\mathscr{C}_i, i\in \{1,\cdots,m\}$.}
    If Assumptions \ref{a1}-\ref{a2} hold, $x_0 \in \bar{\mathscr{C}}$ {\xyr and $B(x) = \mathcal{H}(\psi_{m-1}(x))$ for a continuously differentiable function $\mathcal{H}: \mathbb{R} \rightarrow \mathbb{R}$ satisfying \eqref{eq:ho_safeguard_condition}}, then {\xyr the control policy $u=u^s(x,K_s)$ as defined in $\eqref{e3.2.1a}$ and $\eqref{e3.2.1b}$} renders the set $\bar{\mathscr{C}}$ forward invariant for both undisturbed system \eqref{eq:nominal_sys} and disturbed system \eqref{eq:disturbed_sys}.
\end{lemma}
\begin{proof}
{\xyr
    For the undisturbed system \eqref{eq:nominal_sys}, let $f_{cl} = f$.
    Then $\dot{\mathcal{B}}(x)$ along \eqref{eq:nominal_sys} under the adaptive safeguard policy \eqref{e3.2.1a} is 
    \begin{align*}
    \dot{\mathcal{B}}(x)=&  \frac{\partial \mathcal{B}(x)}{\partial \psi_{m-1}(x)}\Big(\mathcal{L}_{f_{cl}} \psi_{m-1}(x) - (1 - \rho^2 \mu )  K_s \frac{\partial\mathcal{B}(x)}{\partial \psi_{m-1}(x)} \\
   &   \times\mathcal{L}_g \psi_{m-1}(x) R^{-1} \mathcal{L}_g \psi_{m-1} ^{\top}(x) \Big)  \\
     \leq& \left\lvert\frac{\partial \mathcal{B}(x)}{\partial \psi_{m-1}(x)}\right\rvert^2 \bigg(\left\lvert\frac{\partial \mathcal{B}(x)}{\partial \psi_{m-1}(x)}\right\rvert^{-1}\overline{\lVert\mathcal{L}_{f_{cl}}\psi_{m-1}(x)\rVert}_{\mathcal{X}}\\
    & - (1 - \rho^2 \mu )K_s \bar{\sigma}(R)^{-1}  \big(\overline{\lVert \mathcal{L}_g \psi_{m-1}(x)\rVert}_{\mathcal{X}}\big)^2\bigg).
    \end{align*}
    If $(1-\rho^2\mu)K_s > 0$, we can show that $\psi_{m-1}>0$ by following a similar development of the proof as in Theorem \ref{thm:ho_safeguard_safety}.
    Since $\rho^2 \leq 1$ and $\mu < 1$, we have $(1-\rho^2\mu)K_s > 0$ only if $K_s > 0$.
    Then we need to show $K_s > 0$ when $x$ approaches $\partial \mathscr{C}_m$.

    Let $\bar{K}_s =  \mathop{\rm{sup}}_{K_s \in \mathcal{S}} \{K_s\}$. Since $K_s \leq \bar{K}_s$ and the forward invariance of the set $\bar{\mathscr{C}}$ naturally holds when $K_s = \bar{K}_s$ (see Theorem \ref{thm:ho_safeguard_safety}), we only need to focus on the case $K_s < \bar{K}_s$.
    For all $K_s < \bar{K}_s$, the time derivative of $K_s$ along \eqref{e3.2.1b} satisfies
    \begin{align}\label{eq:K_s_dot_lemma}
        \dot{K}_s \geq  
        \gamma e^{-\overline{\lVert h(x)\rVert}_{\mathcal{X}}} \underline{\sigma}(Q) \lVert x\rVert^2- Y K_s^2,~~K_s(0)>0.
    \end{align}
    Using the comparison Lemma, we conclude that $K_s(t) \geq 0$ for all $t\geq 0$.
    Since the term $- Y K_s^2$ decreases monotonically with $K_s$, one has \begin{align}\label{eq:K_s_geq0}
        \dot{K}_s > 0,~~~\forall K_s \in\bigg[0, \sqrt{\frac{\gamma}{Y} e^{-\overline{\lVert h(x)\rVert}_{\mathcal{X}}} \underline{\sigma}(Q)} \lVert x\rVert\bigg).
    \end{align}
    This implies that $K_s > 0$ when $x\in \mathcal{X}\backslash B_r$, where $B_r = \{x\in\mathbb{R}^n|~\lVert x\rVert \leq r\}$ for some $r>0$.
    To see this, suppose by contradiction that there exists $t^*$ such that $x(t^*)\in \mathcal{X}\backslash B_r$ and $K_s(t^*) = 0$.
    Since $K_s$ and $x$ are continuous, there exists a small $\delta > 0$ such that for all $t\in (t^*-\delta, t^*)$, $\lVert x(t) \rVert \geq r$ and $K_s(t)\leq r\sqrt{\frac{\gamma}{2Y} e^{-\overline{\lVert h(x)\rVert}_{\mathcal{X}}} \underline{\sigma}(Q)}$.
    Then for all such $t$, $\dot{K}_s(t) \geq \frac{\gamma}{2} e^{-\overline{\lVert h(x)\rVert}_{\mathcal{X}}} \underline{\sigma}(Q) r^2 > 0$ by \eqref{eq:K_s_dot_lemma}.
    Hence, $K_s(t)$ must be strictly increasing on $(t^*-\delta, t^*)$, which implies $K_s(t)< 0$ when $t \in (t^*-\delta, t^*)$.
    This contradicts with the fact that $K_s(t)\geq 0$.
    Therefore, $K_s>0$ for all $x \in \mathcal{X}\backslash B_r$.

    By Assumption \ref{a2}, the neighborhood $N_l$ of the origin satisfies $N_l \subset \operatorname{Int}(\mathscr{C}_m)$.
    Let $r$ be small enough such that $B_r \subset N_l$.
    %
    %
    We first show that $\psi_{m-1}(x)> 0$ for all $x \in \mathcal{X} \backslash B_r$.
    %
    %
    By Theorem \ref{thm:ho_safeguard_safety}, $K_s>0$ guarantees $\psi_{m-1}(x)>0$, and thus $\psi_{m-1}(x)>0$ for all trajectories on $\mathcal{X}\backslash B_r$.
    Regardless of the value of $K_s$, $\psi_{m-1}(x)>0$ is automatically satisfied for all $x \in B_r$ because $B_r \subset \operatorname{Int}(\mathscr{C}_m)$.

    Now we have shown that $\psi_{m-1}(x)> 0$ for all $x \in \mathcal{X}$.
    It then follows from the proof of Lemma \ref{lm1} that the set $\bar{\mathscr{C}}$ is forward invariant for undisturbed system \eqref{eq:nominal_sys}.
    By redefining $f_{cl} = f+gu^f$, the proof for disturbed system \eqref{eq:disturbed_sys} follows the similar development as that for undisturbed system \eqref{eq:nominal_sys}.
    }
\end{proof}

The following corollary shows that the adaptive safeguard control policy guarantees the safety of system \eqref{eq:disturbed_sys} and its undisturbed system \eqref{eq:nominal_sys} under an existing control policy.

     \begin{corollary} \label{c2}
        {\xyr
        Let $\mathcal{H}: \mathbb{R} \rightarrow \mathbb{R}$ be a continuously differentiable function satisfying \eqref{eq:ho_safeguard_condition} and} $k(x)$ be a locally Lipschitz control policy on $x \in \mathcal{X}$.
        If Assumptions \ref{a1}-\ref{a2} hold, $x_0 \in \bar{\mathscr{C}}$ {\xyr and $B(x) = \mathcal{H}(\psi_{m-1}(x))$, then the controller
        \begin{align} \label{eq:control_policy_safeguard_policy}
            u = k(x)+u^s(x,K_s),
        \end{align}
        where $u^s$ is defined in $\eqref{e3.2.1a}$ and $K_s$ is updated by $\eqref{e3.2.1b}$,
        renders $\bar{\mathscr{C}}$ forward invariant for both system \eqref{eq:nominal_sys} and system \eqref{eq:disturbed_sys}.
        }
        %
        %
    \end{corollary}
    \begin{proof}
        Invoking Lemma \ref{lm:adaptive_safeguard_safety}, we can proceed following the proof of Corollary \ref{corol:us_kx} {\xyr by regarding $g(x)k(x)$ as a bounded disturbance input}. 
    \end{proof}
%

{\xyr
Corollary \ref{c2} shows that the adaptive safeguard policy can be used to modify an existing controller to guarantee safety.}
Suppose that the state trajectories of \eqref{eq:nominal_sys} under the adaptive safeguard policy \eqref{e3.2.1a} and the optimal control policy \eqref{eq:optimal_control_policy} always stay in the admissible safety risk set $\mathscr{B}= \{x \in \mathcal{X}|~  \lVert \nabla \mathcal{B}(x) \rVert^2 \leq \bar{\mathcal{W}}\}$, where $\bar{\mathcal{W}}$ is a positive constant. 
Now we shall provide a sufficient condition to guarantee the stability of the closed-loop system.

\setlength{\itemindent}{0pt}  
\setlength{\leftskip}{0pt}    

\begin{thm} \label{t2}
    Let $V^*$ be a continuously differentiable positive definite solution of \eqref{eq:HJB_without_u} and $\mathscr{B}$ be the admissible safety risk set.
    Let $\rho$ be defined by \eqref{cos}.
    Suppose $x_0 \in {\xyr \bar{\mathscr{C}}}$, Assumption \ref{a2} holds, {\xyr $B(x)$ satisfies the condition in Corollary \ref{c2},} and the optimal control policy \eqref{eq:optimal_control_policy} is locally Lipschitz.
    {\xyr 
    Apply the controller \eqref{eq:control_policy_safeguard_policy} with $k(x)=k^*(x)$ to  the undisturbed system \eqref{eq:nominal_sys}.}
    Then the equilibrium point of the closed-loop system is {\mg asymptotically stable } 
    under either of the following conditions
    \begin{itemize}
        \item[(C1)]  \( 0< \gamma < 1 \), and \( \rho(x) \geq 0, ~  \forall x \in \mathscr{B}\);
        \item[(C2)] $0<\gamma < 1$, $0 \leq \mu < 1$, $Y > \frac{(1-\mu)^2}{1 - \gamma} \frac{\bar{\mathcal{W}} {\xyr \overline{\lVert g(x)\rVert}_{\mathcal{X}}^2}}{\underline{\sigma}(R)}$, and \( \rho(x) < 0 \) for some $x \in \mathscr{B}$.
    \end{itemize}
\end{thm}
\begin{proof}
    Consider the Lyapunov function candidate $L_x(x, K_s) = V^*(x) + K_s$.
    %
    {\xyr Under the control policy \eqref{eq:control_policy_safeguard_policy} with $k(x)=k^*(x)$,}
    the time derivative of $V^*(x)$ along \eqref{eq:nominal_sys} becomes
    \begin{align} \label{e3.2.1}
        \dot{V}^* = \nabla V^*(x)^{\top} \left(f(x) + g(x)(u^s(x,K_s) + k^*(x))\right).
    \end{align}
    {\xyr
    By \eqref{cos} and the definition of $\mathcal{L}_g^+\mathcal{B}$, one can show 
    $\nabla V^{*\top} gR^{-1}\mathcal{L}_g^+\mathcal{B}=0$. 
    Recalling the fact $\mathcal{L}_g\mathcal{B} = \mathcal{L}_g^+\mathcal{B}+\mathcal{L}_g^-\mathcal{B}$, one has $\nabla V^{*\top} gu^s = -(1-\mu) K_s \nabla V^{*\top} gR^{-1}\mathcal{L}_g\mathcal{B}$.
    Further considering \eqref{cos}, we have
    \begin{align}\label{eq:nabla_V_g_us}
    \nabla V^{*\top} gu^s=-(1-\mu) K_s \rho\big \lVert \sqrt{R^{-1}}g^{\top} \nabla V^*\big \rVert \big\lVert \sqrt{R^{-1}}g^\top \nabla \mathcal{B} \big \rVert.
    \end{align}}
    Combining {\xyr \eqref{eq:HJB}}, \eqref{e3.2.1} and {\xyr \eqref{eq:nabla_V_g_us}} yields
    \begin{align*}
        \dot{V}^* = & -x^{\top} Q x - {\xyr\frac{1}{4} \big\lVert \sqrt{R^{-1}}g^{\top} \nabla V^* \big\rVert^2} \\
        & -{\xyr (1-\mu) K_s \rho \big\lVert \sqrt{R^{-1}}g^{\top} \nabla V^* \big\rVert \big\lVert \sqrt{R^{-1}} g^\top \nabla \mathcal{B} \big\rVert}.
    \end{align*}
    According to Theorem 3.10.1 in \cite{ioannou2006adaptive}, the time derivative of $L_x$ along \eqref{eq:nominal_sys} and \eqref{e3.2.1b} is 
    \begin{align*}
         \dot{L}_x  \leq & -(1-\gamma) \left(x^{\top} Q x + \frac{1}{4} \big\lVert \sqrt{R^{-1}}g^{\top} \nabla V^* \big\rVert^2 \right)-YK_s^2\\
         &  -{\xyr  (1-\mu) K_s \rho(x) \big\lVert \sqrt{R^{-1}}g^{\top} \nabla V^* \big\rVert \big\lVert \sqrt{R^{-1}}g^\top \nabla \mathcal{B} \big\rVert}.
    \end{align*}

    \begin{itemize}
        \item[(1)] 
        Under condition C1 and the fact that $K_s(t) \geq 0$ by Lemma \ref{lm:adaptive_safeguard_safety}, the time derivative of $V^*$ is  
        \begin{align*}
            \dot{V}^* \leq  -x^{\top} Q x - \frac{1}{4} \big\lVert \sqrt{R^{-1}}g^{\top} \nabla V^* \big\rVert^2,
        \end{align*}
        and $\dot{L}_x$ is given as
        \begin{align*}
            \dot{L}_x \leq  -(1- \gamma) \left(x^{\top} Q x + \frac{1}{4} \big\lVert \sqrt{R^{-1}}g^{\top} \nabla V^* \big\rVert^2 \right)- YK_s^2.
        \end{align*}
        Hence, the Lyapunov stability condition is always satisfied when $\rho \geq 0$.
        \item[(2)]
        Under condition C2, if a constant $K_s$ is adopted, the stability of the closed-loop system might be destroyed due to the conflict between $u^s$ and $k^*$, i.e., $\rho <0$.
        However, the stability condition of $L_x$ can still be guaranteed by the adaptive law \eqref{e3.2.1a}, \eqref{e3.2.1b}.
        The time derivative of $L_x$ is
        \begin{align*}
            \dot{L}_x \leq&  -(1 - \gamma) \left( x^{\top} Q x + \frac{1}{4} \big\lVert \sqrt{R^{-1}}g^{\top} \nabla V^* \big\rVert^2 \right) - YK_s^2   \\
            & -(1 - \mu) \rho\big \lVert \sqrt{R^{-1}}g^{\top} \nabla V^* \big\rVert\big \lVert \sqrt{R^{-1}}g^{\top} \nabla \mathcal{B} \big\rVert K_s.
        \end{align*}
        Completing the squares and considering the fact that $\lVert \rho \rVert \leq 1$, one further has 
        \begin{align*}
            \dot{L}_x \leq & -(1 - \gamma) \left(x^{\top} Q x  + \frac{\big\lVert \sqrt{R^{-1}}g^{\top} \nabla V^* \big\rVert^2}{4Y}\big( Y \right.\\
             & \left. -  \frac{(1 -\mu)^2}{1-\gamma}\big\lVert \sqrt{R^{-1}}g^{\top} \nabla \mathcal{B} \big\rVert^2 \big)\right).
        \end{align*}
        {\xyr 
        Combining $\lVert g R g^\top \rVert \leq  \underline{\sigma}(R)^{-1} \overline{\lVert g(x)\rVert}_{\mathcal{X}}^2,$
        $(1-\gamma)Y > (1-\mu)^2\frac{\bar{\mathcal{W}} \overline{\lVert g(x)\rVert}_{\mathcal{X}}^2}{\underline{\sigma}(R)}$} and $\gamma < 1$, we have $\dot{L}_x \leq -(1-\gamma)x^{\top} Q x$, for any $x\in \mathscr{B}$.
    \end{itemize}
    %
    %
    Therefore, under both conditions C1 and C2, one has $L_x(t) \leq L_x(0), \forall t \geq 0$.
    According to LaSalle’s invariance principle, one can conclude that ${\xywr \lim}_{t\rightarrow \infty} x(t) = 0$.
    Hence, asymptotic stability of the system \eqref{eq:nominal_sys} {\mg is achieved} under {\xyr the control policy \eqref{eq:control_policy_safeguard_policy} with $k(x)=k^*(x)$}.
\end{proof}

\section{Implementation of robust safe RL} \label{Sec-PF}
This section develops a robust safe RL framework for learning $V^*(x)$ and $k^*(x)$ {\xyr based on the nominal model \eqref{eq:nominal_sys}},  while guaranteeing safety and robustness.
%

\subsection{Disturbance/fault observer design}
To compensate for unknown disturbance/fault $u^f$, the following nonlinear disturbance/fault observer is adopted \footnote{{\xyr This observer has been widely used in the robust control literature \cite{ersin2025safety, xie2021disturbance}, and its detailed derivation is omitted for brevity.}}
\begin{align} \label{eq:disturbance_observer}
&\hat{u}^f(t)  =  z(t) + \omega(x) \nonumber, \\
&\dot{z}(t) = -L(x)\left(f(x)+g(x)\left( u(t) + z(t) + \omega(x) \right) \right),
\end{align}
where $\hat{u}^f$ is the estimation of $u^f$, $z$ is an auxiliary vector, $L(x) = (\partial \omega(x) / \partial x)^{\top}$ and $\omega(x)$ is a user-defined function.

Let the estimation error be $\Tilde{u}^f = u^f - \hat{u}^f$.
{\xyr Then}
\begin{align*}
    \dot{\Tilde{u}}^f = -L(x) g(x) \Tilde{u}^f + \dot{u}^f.
\end{align*}
%
%
{\mg
An appropriate $\omega(x)$ is chosen} such that 
{\xyr $L(x) g(x) \geq K_f I$ for a $K_f>0$.
\footnote{{\xyr If $g(x)$ is full column rank, then one can select $\omega$ by solving $\frac{\partial \omega(x)}{\partial x}=K_fg(x)^{\dagger}$, where $g(x)^{\dagger}$ denotes the left inverse of $g(x)$.}}
} 
{\xyr Consider a Lyapunov function $L_u(\tilde{u}^f) = \frac{1}{2} \tilde{u}^{f\top}\tilde{u}^f$, whose time derivative satisfies
\begin{align} \label{eq:do_stable_condition}
    \dot{L}_u(\tilde{u}^f) \leq -K_f \lVert \tilde{u}^f \rVert^2 + \lVert \dot{u}^f(t) \rVert_\infty \lVert \tilde{u}^f \rVert.
\end{align}
%
Following a similar development as in Lemma 1 of \cite{ersin2025safety}, the estimation error is bounded as
\begin{align}\label{eq:til_u_bound}
    \lVert \tilde{u}^f(t) \rVert\leq e^{-K_ft} \lVert \tilde{u}^f(0) \rVert + \frac{1}{K_f}\lVert \dot{u}^f(t)\rVert_\infty \leq \bar{\tilde{u}}^f,
\end{align}
}
where $\bar{\tilde{u}}^f = {\xyr \lVert \dot{u}^f (t)\rVert_\infty}/{K_f} +\lVert \tilde{u}^f{\xyr(0)}\rVert$.

{\xyr 
The following corollary shows that Assumption \ref{a1} can be relaxed to unbounded disturbance/fault by using observer \eqref{eq:disturbance_observer}, if the disturbance satisfies
    \begin{equation} \label{uf-relaxed}
        \lVert \dot{u}^f(t) \rVert_\infty \leq \tilde{\eta},~~ \lVert u^f(0)\rVert < \infty,
    \end{equation}
where $\tilde{\eta}$ is an unknown positive constant.

\begin{corollary} \label{corol:us_uf}
    Consider the observer \eqref{eq:disturbance_observer} and the safeguard policy \eqref{eq:ho_safeguard}.
    Let $\mathcal{H}:\mathbb{R} \rightarrow \mathbb{R}$ be a continuously differentiable function satisfying \eqref{eq:ho_safeguard_condition}.
    Suppose Assumption \ref{a2} and the condition \eqref{eq:do_stable_condition} hold. 
    If $x_0 \in \bar{\mathscr{C}}$,  $B(x) = \mathcal{H}(\psi_{m-1}(x))$ and $u^f$ satisfies \eqref{uf-relaxed}, then $u = u^s(x) - \hat{u}^f$ renders the set $\bar{\mathscr{C}}$ forward invariant for  system \eqref{eq:disturbed_sys}.
\end{corollary}
\begin{proof}
    The closed-loop system is
    \begin{align} \label{eq:cls_us_uf}
        \dot{x} = f(x)+g(x)u^s(x) + g(x)\tilde{u}^f.
    \end{align}
    Since \eqref{eq:do_stable_condition} and \eqref{uf-relaxed} hold, the bound $\bar{\tilde{u}}^f$ is finite.
    Redefining $f_{cl}(x,t) = f(x)+ g(x)\tilde{u}^f$, one has $$\lVert f_{cl}(x,t) \rVert\leq \overline{\lVert f(x) \rVert}_{\mathcal{X}} + \overline{\lVert g(x) \rVert}_{\mathcal{X}}\bar{\tilde{u}}^f<\infty.$$
    It then follows a similar development as in the proof of Theorem \ref{thm:ho_safeguard_safety} that $\bar{\mathscr{C}}$ is forward invariant for system \eqref{eq:cls_us_uf}.
\end{proof}

Following a similar analysis in the proof of Corollary \ref{corol:us_uf}, the results in Corollary \ref{corol:us_kx} and Lemma \ref{lm4} can be also extended to unbounded disturbance/fault $u^f$ satisfying \eqref{uf-relaxed} by using $\hat{u}^f$ for disturbance/fault compensation.

\begin{remark} 
    The continuous differentiability of $u^f$ is required for observer design.
    This property \eqref{eq:til_u_bound} ensures that $u^f$ can be estimated and compensated, thereby enhancing the closed-loop performance, and, as shown in \cite{ersin2025safety}, leading to a less conservative CBF constraint.
    However, it is important to note that forward invariance of the set $\bar{\mathscr{C}}$ can still be preserved as long as $u^f$ remains bounded, even if it is discontinuous (see Theorem \ref{thm:ho_safeguard_safety}). 
    In such non-differentiable cases, input-to-state stability can still be  guaranteed; however, 
    the control performance may be degraded. 
    Specifically, 
    without disturbance compensation, by applying Theorem 8.3 in \cite{cohen2023adaptive}, we can show that $\lim\sup_{t\rightarrow\infty} \lVert x(t)\rVert =\breve{\alpha}(\lVert u^f(t)\rVert_\infty)$ for a class $\mathcal{K}$ function $\breve{\alpha}$; with disturbance compensation, we have $\lim\sup_{t\rightarrow\infty} \lVert x(t)\rVert =\breve{\alpha}\left(\frac{\lVert \dot{u}^f(t)\rVert_\infty}{K_f}\right)$ by \eqref{eq:til_u_bound}.
    Since $\breve{\alpha}\left(\frac{\lVert \dot{u}^f(t)\rVert_\infty}{K_f}\right)$ can be made arbitrary small by selecting a large $K_f$, the state trajectory with disturbance compensation will converge to a much smaller value than that without disturbance compensation. 
\end{remark}
}

\subsection{Actor-critic neural networks}
Over the compact set $\mathcal{X}$, the optimal value function $V^*(x)$ can be  approximated as
 $$  V^*(x) = W^{\top} \phi(x) + \varepsilon(x),$$
where $W \in \mathbb{R}^s$ is the ideal weight, $\phi:\mathbb{R}^n \rightarrow \mathbb{R}^s$ is the activation function with $\phi(0) = 0$ and $\nabla \phi(0) = 0$, and $\varepsilon: \mathbb{R}^n \rightarrow \mathbb{R}$ is the approximation error.
By the universal function approximation property \cite{FL99BookNN}, the approximation error $\varepsilon(x)$ can be made arbitrarily small and bounded by a given positive constant $\bar{\varepsilon}$. 

The derivative of $V^*(x)$ along the state trajectory is
$$\nabla V^*(x) = \nabla \phi(x)^{\top} W + \nabla \varepsilon(x),$$
where $\nabla \varepsilon(x)$ is bounded by  some positive constant $\bar{\varepsilon}_{\nabla}$ over the compact set $\mathcal{X}$. 
Since the ideal weight $W$ is unknown, a critic NN \eqref{e4.2.1} and an actor NN \eqref{e4.2.2} are used to approximate the optimal value function and the optimal control policy
\begin{align}
     \hat{V}(x) &= \hat{W}_c^{\top} \phi(x),   \label{e4.2.1} \\
     K(x, t) &= -\frac{1}{2}R^{-1} g(x)^{\top} \nabla \phi(x)^{\top} \hat{W}_a(t) \label{e4.2.2},
\end{align}
where $\hat{W}_c$ is the critic NN weight and $\hat{W}_a$ is the actor NN weight.

\subsection{ Model-based safe RL via simulation of experience}
This section develops a safe robust RL framework for learning the optimal value function and optimal control policy based on simulation of experience.
The behavior policy is defined as 
\begin{align} \label{e4.3.1}
    u = K(x, t) + u^s(x, K_s) - \hat{u}^f,
\end{align}
where the first term is the target policy, the second term is the safeguard policy, and the last term is the estimated disturbance/fault.
{\xyr
A performance metric for learning the ideal NN weights is the Bellman error, which is defined as $\delta= H(x,\nabla \hat{V}, u) - H(x,\nabla V^*,k^*)$.
By \eqref{eq:HJB}, the Bellman error is}
\begin{align} \label{e4.3.2}
\delta=\hat{W}_c^{\top} \varphi + l(x, u),
\end{align}
where $\varphi = \nabla \phi(x) (f(x) + g(x)u)$ and $\nabla \phi = \partial \phi/ \partial x$.
Given that online learning is actually off-policy since the behavior policy \eqref{e4.3.1} is different from $K$, the collected data may result in an approximation error in the optimal control policy \eqref{eq:optimal_control_policy}.
Hence, the undisturbed system \eqref{eq:nominal_sys} is utilized to generate simulation data to correct the approximation error and relax the PE condition as well.
Denote $\{x_i(t)\}_{i=1}^N$ as a collection of $N$ state points sampled from $\mathcal{X}$ at $t \geq 0$ and {\xywv $u_i = K(x_i, t)$}
as the exploratory policy.
The Bellman error at each sampled state is obtained as
\begin{align} \label{e4.3.3}
    \delta_i = \hat{W}_c^{\top} \varphi_i + l(x_i, u_i), 
\end{align}
where $\varphi_i = \nabla \phi(x_i) (f(x_i) + g(x_i)u_i)$.
The tuning law for the critic NN weights is 
\begin{align} 
    \dot{\hat{W}}_c &= -\Gamma \Big(k_{c1} \frac{\varphi}{{\xyr \vartheta}^2} \delta + \frac{k_{c2}}{N}\sum_{i=1}^N \frac{\varphi_i}{{\xyr \vartheta}_i^2(t)} \delta_i\Big), \label{e4.3.4a} \\
    \dot{\Gamma}& = \beta \Gamma - \Gamma \Big( k_{c1} \Lambda + \frac{k_{c2}}{N}\sum_{i=1}^N \Lambda_i \Big) \Gamma, \label{e4.3.4b}
\end{align}
where ${\xyr \vartheta}(t) = 1 + \varphi(t)^{\top} \varphi(t)$, ${\xyr \vartheta}_i(t) = 1 + \varphi_i(t)^{\top} \varphi_i(t)$,
$$\Gamma(0) > 0,\ \Lambda(t) = \frac{\varphi(t) {\varphi(t)}^{\top}}{{\xyr \vartheta}^2(t)}\ \ \textnormal{and}\ \  \Lambda_i(t) = \frac{\varphi_i(t) {\varphi_i(t)}^{\top}}{{\xyr \vartheta}_i^2(t)},$$
are normalizing signals, $k_{c1},k_{c2}> 0$ are learning gains, and $\beta>0$ is a forgetting factor. 
Let $G_{\phi} = \nabla \phi(x) G_R(x) \nabla \phi(x)^{\top}$  \textnormal{and} $  G_{\phi, i} = \nabla \phi(x_i) G_R(x_i) \nabla \phi(x_i)^{\top}$, where $G_R(x)=g(x)R^{-1}g(x)^\top$.
The tuning law of the actor NN weights is 
\begin{align} \label{e4.3.5}
    \dot{\hat{W}}_a = & \operatorname{Proj}_{\mathcal{W}} \biggl\{-k_{a1}\left(\hat{W}_a - \hat{W}_c \right)  +  \frac{k_{c1}}{4 {\xyr \vartheta}^2} G_{\phi} \hat{W}_a \varphi^{\top} \hat{W}_c \nonumber \\
    & - k_{a2} \hat{W}_a + \sum_{i=1}^N \frac{k_{c2}}{4 N{\xyr \vartheta}_i^2} G_{\phi, i}\hat{W}_a \varphi_i^{\top}  \hat{W}_c \biggl\},
\end{align}
where $k_{a1}, k_{a2} > 0$ are learning gains, and $\operatorname{Proj}_{\mathcal{W}} \{\cdot \}$ is a projection operator for keeping the actor NN estimation within a compact set $\mathcal{W} \in \mathbb{R}^s$.

To provide a PE-like condition {\mg that can be easily {\hwr checked} during the learning process,} 
and facilitate the stability analysis, we make the following assumption.
%
\begin{assumption} \label{a3}
Given the NNs \eqref{e4.2.1} and \eqref{e4.2.2}, the following conditions hold: \footnote{Assumption \ref{a3} is standard in solving optimal control problem using RL \cite{zhang2014online}.
    The first assumption is made to exclude the non-existence solution. 
    The second assumption can be satisfied by taking sigmoid, Gaussian and other standard NN activation functions \cite{vamvoudakis2010online}. 
    The last condition can be satisfied by choosing a large number of sampled data \cite{kamalapurkar2016model}.}
\begin{enumerate} 
    \item The optimal weight $W$ is norm-bounded by an unknown positive constant $\bar{W}$.
    \item The activation function $\phi(x)$ and its derivative $\nabla \phi(x)$ are norm-bounded by positive constants $\bar{\phi}$ and $\bar{\phi}_\nabla$ {\xywv on $\mathcal{X}$}, respectively. 
    \item The sample number $N$ is large enough to satisfy the following PE-like condition 
       \begin{align*}
        \inf_{t\geq 0} \left\{\underline{\sigma} \left( \frac{1}{N} \sum_{i=1}^N \Lambda_i(t) \right)\right\} \geq  \lambda_c,
       \end{align*}
    where $\lambda_c$ is a positive constant.
\end{enumerate}
\end{assumption}


\subsection{Safety and stability analysis}
In this section, we analyze the safety and stability of the closed-loop system \eqref{eq:disturbed_sys} under the control policy \eqref{e4.3.1}, i.e., 
\begin{align}\label{cls-eq}
        \dot{x}= f(x) + g(x) (K(x, t) + u^s(x, K_s) {\xyr + \tilde{u}^f}).
\end{align}
\begin{lemma} \label{lm4}
    Suppose Assumptions \ref{a2}-\ref{a3}, condition \eqref{uf-relaxed} and inequality \eqref{eq:do_stable_condition} hold.
    {\xyr Let $\mathcal{H}: \mathbb{R} \rightarrow \mathbb{R}$ be a continuously differentiable function satisfying \eqref{eq:ho_safeguard_condition} and the control policy be \eqref{e4.3.1}, where $u^s$ is defined in \eqref{e3.2.1a} with $K_s$ updated by \eqref{e3.2.1b}, $\hat{u}^f$ is estimated by \eqref{eq:disturbance_observer} and $K(x,t)$ is defined in \eqref{e4.2.2} with NN weight updated by \eqref{e4.3.5}.
    If $x_0 \in \bar{\mathscr{C}}$ and $B(x) =\mathcal{H}(\psi_{m-1}(x))$}, then the set $\bar{\mathscr{C}}$ is forward invariant for system \eqref{cls-eq}.
\end{lemma}
\begin{proof}
    {\xyr
    By Assumption \ref{a3}, one has $$\lVert g(x)^{\top} \nabla \phi(x)^{\top} \rVert \leq \overline{\lVert g(x)\rVert}_{\mathcal{X}} \bar{\phi}_\nabla.$$
    Further, \eqref{e3.2.1b} guarantees $\lVert \hat{W}_a\rVert \leq \bar{W}_a$, and thus $$\lVert K(x,t)\rVert \leq\frac{1}{2}\underline{\sigma}(R)^{-1}  \overline{\lVert g(x)\rVert}_{\mathcal{X}}\bar{\phi}_\nabla \bar{W}_a.$$
    By \eqref{eq:do_stable_condition} and \eqref{uf-relaxed}, the bound $\bar{\tilde{u}}^f$ is finite.
    Redefine $f_{cl}(x,t) = f(x) + g(x) (K(x,t)+\tilde{u}^f)$, and then 
    \begin{align*}
        \lVert f_{cl}(x,t)\rVert \leq& \overline{\lVert f(x)\rVert}_{\mathcal{X}} +\overline{\lVert g(x)\rVert}_{\mathcal{X}}\\
    &\times \left(\frac{1}{2}\underline{\sigma}(R)^{-1} \overline{\lVert g(x)\rVert}_{\mathcal{X}}\bar{\phi}_\nabla \bar{W}_a + \bar{\tilde{u}}^f\right).
    \end{align*}
    %
    %
    Since $\overline{\lVert f(x)\rVert}_{\mathcal{X}}$, $\overline{\lVert g(x)\rVert}_{\mathcal{X}}$, $\bar{\tilde{u}}^f$ and $\bar{W}_a$ are bounded, we have that $\lVert f_{cl}(x,t)\rVert$ is bounded.
    It then follows the similar developments as in the proof of Theorem \ref{thm:ho_safeguard_safety} that the set $\bar{\mathscr{C}}$ is forward invariant for system \eqref{cls-eq}.
    }
\end{proof}

%
%
Define the weight estimation errors as $\Tilde{W}_c = W - \hat{W}_c \ \ \textnormal{and} \ \ \Tilde{W}_a = W - \hat{W}_a,$  and a composite state $z = \col\bm(x, \Tilde{W}_c, \Tilde{W}_a, \Tilde{u}^f, K_s \bm)$.
To guarantee that the value function approximation error $\varepsilon(x)$ and its derivative $\nabla \varepsilon(x)$ are norm-bounded during the learning process, one can construct a compact set $\mathcal{Z} \subset \mathcal{X} \times \mathbb{R}^{2s+p+1}$ that contains the composite state trajectory.\footnote{Readers can be referred to \cite{kamalapurkar2016model} for more details in computing such a compact set.}
Define $\bar{\varepsilon} = \overline{\lVert \varepsilon (x)  \rVert}_\mathcal{X}$ and $\bar{\varepsilon}_\nabla = \overline{\lVert \nabla \varepsilon(x) \rVert}_\mathcal{X}$.
The following theorem illustrates that the robust and adaptive safe {\xyr RL framework} ensures the composite state {\mg is}
uniformly ultimately bounded.
%
\begin{thm} \label{t3}
    {\xyr
    Consider the closed-loop system \eqref{cls-eq}.}
    Let the safeguard controller and its adaptive law be given by \eqref{e3.2.1a} and  \eqref{e3.2.1b}, the disturbance/fault observer be constructed as \eqref{eq:disturbance_observer}, $K(x,t)$ be defined as \eqref{e4.2.2} with the critic NN updated by \eqref{e4.3.4a} and \eqref{e4.3.4b}, and the actor NN updated by \eqref{e4.3.5}, respectively.
    %
    %
    Suppose that Assumptions \ref{a2}-\ref{a3}, condition \eqref{uf-relaxed} and inequality \eqref{eq:do_stable_condition} hold.
    {\xyr If $x_0 \in \bar{\mathscr{C}}$ and $B(x)$ satisfies the conditions in Lemma \ref{lm4}}, then the control policy \eqref{e4.3.1} guarantees that the composite state $z(t)$ is uniformly ultimately bounded.
\end{thm}

\begin{proof}
Consider the following Lyapunov function
\begin{align} \label{ea.1}
    \hspace{-0.1in}P(z, t) = L_x(x, K_s) + L_c(\Tilde{W}_c, t) + L_a(\Tilde{W}_a) + L_u(\Tilde{u}^f),
\end{align}
where $L_c(\Tilde{W}_c, t) = \frac{1}{2} \Tilde{W}_c^{\top} \Gamma^{-1}(t) \Tilde{W}_c$, $L_a(\Tilde{W}_a) = \frac{1}{2} \Tilde{W}_a^{\top} \Tilde{W}_a$ and $L_u(\Tilde{u}^{f})$ is a positive definite function satisfying \eqref{eq:do_stable_condition}.

Under Assumption \ref{a3}, there exist positive constants $\bar{\Gamma}$ and $\underline{\Gamma}$ such that $\underline{\Gamma} I\leq  \Gamma(t) \leq \bar{\Gamma} I$, $\forall t \geq 0$ (\cite{rushikesh2016efficient}, Lemma 1).
Therefore, $P(z, t)$ is positive definite and hence $$\eta_1(\lVert z \rVert) \leq P(z, t) \leq \eta_2(\lVert z \rVert)$$ for some class $\mathcal{K}$ functions $\eta_1$ and $\eta_2$ \cite{khalil2002nonlinear}.

The time derivative of $V^*$ along the system \eqref{cls-eq} is
\begin{align*}
    \dot{V}^*  = \nabla {V^*}^{\top} \left(f(x)+g(x)(K(x,t) + u^s(x,K_s)+ \Tilde{u}^f)\right).
\end{align*}
Noting that $H(x, \nabla V^*, k^*) = 0$, one has 
\begin{align}\label{ea.2}
    \dot{V}^* = -x^{\top} Q x - {k^*}^{\top} R k^*+ \nabla {V^*}^{\top} g(x)(u^s + \Tilde{u}^f - \Tilde{k})
\end{align}
where $\Tilde{k} = k^* - K(x,t)
= -\frac{1}{2}R^{-1}g(x)^{\top} ( \nabla \phi^{\top} \Tilde{W}_a + \nabla \varepsilon).$


Taking \eqref{e3.2.1a}, \eqref{e3.2.1b} and \eqref{ea.2} into consideration, the time derivative of $L_x$ is
\begin{align} \label{ea.3}
    \dot{L}_x \leq & -(1 - \gamma)x^{\top} Q x + \frac{1}{2}W^{\top} G_{\phi} \Tilde{W}_a - YK_s^2 + \frac{\gamma}{4} \hat{W}_a^{\top} G_{\phi} \hat{W}_a \nonumber \\
    & +\frac{1}{2}\nabla \varepsilon^{\top} G_R \nabla \phi^{\top} \Tilde{W}_a + (W^{\top} \nabla \phi +\nabla \varepsilon^{\top}) g(x) \Tilde{u}^f \nonumber\\
    & -(W^{\top} \nabla \phi + \nabla \varepsilon^{\top}) G_R \nabla \mathcal{B} K_s   + \bar{\varepsilon}_1 \nonumber\\
    \leq & -(1 - \gamma)x^{\top} Q x - YK_s^2  + \frac{\gamma}{4}  \bar{G}_{\phi} \lVert W \rVert^2 + \frac{\gamma}{4} \bar{G}_{\phi} \lVert \tilde{W}_a \rVert^2  \nonumber\\ 
    &  + \frac{l_1}{2} \lVert \Tilde{W}_a \rVert  + l_2 \bar{g} \lVert \Tilde{u}^f \rVert  + l_2  \bar{K}_s \bar{G}_R \lVert \nabla \mathcal{B} \rVert_{\infty} + \bar{\varepsilon}_1, 
\end{align}
where $\bar{g}=\overline{\lVert g(x)\rVert}_\mathcal{X},$  
\begin{align*}
& G_\varepsilon = \nabla \varepsilon^{\top} G_R \nabla \varepsilon,~\bar{G}_R = \underline{\sigma}(R)^{-1}\bar{g}^2,~\bar{G}_{\phi} = \bar{G}_R \bar{\phi}_{\nabla}^2,\\
&\bar{\varepsilon}_1  = \frac{1}{2}\overline{\lVert W^{\top} \nabla \phi G_R \nabla \varepsilon + G_\varepsilon \rVert}_{\mathcal{X}},\\
& l_1 = (1-\gamma)\bar{W} \bar{G}_{\phi} + \bar{\varepsilon}_{\nabla} \bar{G}_R \bar{\phi}_{\nabla}\quad \text{and} \quad l_2 = \bar{W} \bar{\phi}_{\nabla} + \bar{\varepsilon}_{\nabla}.
\end{align*}
%



The time derivative of $L_c$ is 
\begin{align} \label{ea.4}
     \dot{L}_c = \Tilde{W}_c^{\top} \Gamma^{-1} \dot{\Tilde{W}}_c - \frac{1}{2} \Tilde{W}_c^{\top} \Gamma^{-1} \dot{\Gamma} \Gamma^{-1} \Tilde{W}_c,
\end{align}
where 
\begin{align}\label{ea.5}
    \dot{\Tilde{W}}_c =  \Gamma \left(k_{c1} \frac{\varphi}{{\xyr \vartheta}^2} \delta + \frac{k_{c2}}{N}\sum_{i=1}^N \frac{\varphi_i}{{\xyr \vartheta}_i^2} \delta_i\right).
\end{align}
Based on the fact that $$K^{\top} R K = \frac{1}{4}W^{\top} G_{\phi} W + \frac{1}{4} \Tilde{W}_a^{\top} G_{\phi} \Tilde{W}_a - \frac{1}{2}W^{\top} G_{\phi} \Tilde{W}_a$$ and $H(x, \nabla V^*, k^*) = 0$, one has 
\begin{align*}
     W^{\top} \nabla \phi & f(x) + x^{\top} Q x - \frac{1}{4}W^{\top} G_{\phi} W \\
    &= \frac{1}{4} \varepsilon^{\top} G_R \varepsilon + \frac{1}{2} \nabla \varepsilon^{\top} G_R \nabla \phi^{\top} W - \nabla \varepsilon^{\top} f(x)
\end{align*}
and then the Bellman error \eqref{e4.3.2} can be rewritten as
\begin{align} \label{ea.6}
    \delta = & - \varphi^{\top} \Tilde{W}_c + x^{\top} Q x + u^{\top} R u + \varphi^{\top} W  \nonumber \\
    = & - \varphi^{\top} \Tilde{W}_c + \frac{1}{4} \Tilde{W}_a^{\top} G_{\phi} \Tilde{W}_a + \varepsilon_{H} + \xi,
\end{align}
where  $\varepsilon_H = \frac{1}{4} \varepsilon^{\top} G_R \varepsilon + \frac{1}{2} \nabla \varepsilon^{\top} G_R \nabla \phi^{\top} W - \nabla \varepsilon^{\top} f(x)$ and 
\begin{align*}
    \xi =&  - \Tilde{W}_a^{\top} \nabla \phi G_R \nabla \mathcal{B} K_s - \Tilde{W}_a^{\top} \nabla \phi g(x) u^f + \Tilde{W}_a^{\top} \nabla \phi g(x) \Tilde{u}^f \\
    & + \nabla \mathcal{B}^{\top} G_R \nabla \mathcal{B} K_s^2 +2 \nabla \mathcal{B}^{\top} g(x) u^f K_s - 2 K_s \nabla \mathcal{B}^{\top} g(x) \Tilde{u}^f \\
    & + u^{fT} R u^f + \Tilde{u}^{fT} R \Tilde{u}^f - 2 u^{fT} R \Tilde{u}^f.
\end{align*}
Similarly, the Bellman error \eqref{e4.3.3} estimated at $x_i$, sampled from the compact set $\mathcal{X}$, can be obtained as
\begin{align} \label{ea.7}
    \delta_i & = \hat{W}_c^{\top} \nabla \phi_i (f_i(x) + g_i(x) K_i) + x_i^{\top} Q x_i + K_i^{\top} R K_i, \nonumber \\
    &  = - \varphi_i^{\top} \Tilde{W}_c + \frac{1}{4} \Tilde{W}_a^{\top} G_{\phi, i} \Tilde{W}_a + \varepsilon_{H, i},
\end{align}
where $\varepsilon_{H, i} = \varepsilon_{H}(x_i)$ and $K_i = K(x_i)$.

Combining \eqref{ea.5}, \eqref{ea.6} and \eqref{ea.7}, \eqref{ea.4} can be put as
\begin{align} \label{ea.8}
    \dot{L}_c = & -\frac{1}{2} \Tilde{W}_c^{\top} \left( \beta \Gamma^{-1} + k_{c1} \Lambda + \frac{k_{c2}}{N}\sum_{i=1}^N \Lambda_i \right) \Tilde{W}_c \nonumber \\
    &+ \Tilde{W}_c^{\top} \left(  k_{c1} \frac{\varphi}{{\xyr \vartheta}^2} \varepsilon_{H} + \frac{k_{c2}}{N}\sum_{i=1}^N \frac{\varphi_i}{{\xyr \vartheta}_i^2} \varepsilon_{H, i} \right) \nonumber \\
    & + \Tilde{W}_c^{\top} k_{c1} \frac{\varphi}{{\xyr \vartheta}^2} \xi + \frac{1}{4} \Tilde{W}_c^{\top} \varrho,
\end{align}
where $ \varrho = k_{c1} \frac{\varphi}{{\xyr \vartheta}^2} \Tilde{W}_a^{\top} G_{\phi} \Tilde{W}_a + \frac{k_{c2}}{N}\sum_{i=1}^N \frac{\varphi_i}{{\xyr \vartheta}_i^2} \Tilde{W}_a^{\top} G_{\phi, i} \Tilde{W}_a.$

Since $f(x)$ is locally Lipschitz, its Lipschitz constant is bounded in the compact set $\mathcal{X}$.
Define $\bar{\varepsilon}_H = \overline{\lVert \varepsilon_H\rVert}_{\mathcal{X}}$.
Using the fact that the inequality $\left\lVert \frac{w}{(1+w^{\top} w)^2} \right\rVert \leq \frac{3\sqrt{3}}{16}$ always holds for any vector $w \in \mathbb{R}^q$, we further have 
\begin{align} \label{ea.9}
    \dot{L}_c \leq & -k_{c2} \bar{\lambda}_c \lVert \Tilde{W}_c\rVert^2 + \frac{3\sqrt{3}}{16}\bar{k}_c \bar{\varepsilon}_H \lVert \Tilde{W}_c \rVert \nonumber \\
    & + \Tilde{W}_c^{\top} k_{c1} \frac{\varphi}{{\xyr \vartheta}^2} \xi + \frac{1}{4} \Tilde{W}_c^{\top} \varrho,
\end{align}
where $\bar{\lambda}_c = \frac{\beta}{2 k_{c2} \bar{\Gamma}} + \frac{\lambda_c}{2}$ and $\bar{k}_c = k_{c1}+k_{c2}$.

Since the projector in \eqref{e4.3.5} guarantees $\lVert \hat{W}_a \rVert \leq \bar{W}_a$ for some $\bar{W}_a > 0$ in a compact set $\mathcal{W}$, one has {\xyr$\lVert \Tilde{W}_a \rVert \leq \lVert \hat{W}_a \rVert + \lVert W \rVert,$} i.e., $\Tilde{W}_a$ is norm-bounded by $\bar{W}_a + \bar{W}$.
Due to the fact that  $\Tilde{W}_a$ and $\Tilde{u}^f$ are norm-bounded, the inequality $\dot{L}_x$ \eqref{ea.3} implies that $x$ and $K_s$ are norm-bounded.
Hence, one can obtain the following inequalities
\begin{align}\label{ea.10}
    \Tilde{W}_c^{\top} k_{c1} \frac{\varphi}{{\xyr \vartheta}^2} \xi \leq k_{c1} \frac{3\sqrt{3}}{16} \lVert \Tilde{W}_c \rVert \lVert \xi \rVert
\end{align}
and 
\begin{align}\label{ea.11}
    \lVert \xi \rVert \leq & \left( \bar{K}_s \bar{G}_R \lVert \nabla \mathcal{B} \rVert_{\infty}  +  (\eta_1 + \bar{\Tilde{u}}^f)  \bar{g}  \right) \bar{\phi}_\nabla \lVert \Tilde{W}_a \rVert \nonumber\\
    & + 2 (\eta_1 \bar{\sigma}(R) + \bar{g} \bar{K}_s \lVert \nabla \mathcal{B} \rVert_{\infty}) \lVert \Tilde{u}^f \rVert + \bar{K}_s^2 \bar{G}_R \lVert \nabla \mathcal{B} \rVert_{\infty}^2 \nonumber\\
    &  + 2 \eta_1 \bar{g} \lVert \nabla \mathcal{B} \rVert_{\infty}  \lVert K_s \rVert   + \bar{\sigma}(R) \left(\eta_1^2 + (\bar{\Tilde{u}}^{f})^2 \right).
\end{align}

Similarly, taking the time derivative of $L_a$ along \eqref{e4.3.5} yields
\begin{align} \label{ea.12}
    \dot{L}_a \leq & -\bar{k}_a \Tilde{W}_a^{\top}\Tilde{W}_a + k_{a1}\Tilde{W}_a^{\top} \Tilde{W}_c + k_{a2} \Tilde{W}_a^{\top} W \nonumber\\
    & -  k_{c1} \Tilde{W}_a^{\top} \frac{G_{\phi}}{4 {\xyr \vartheta}^2} \hat{W}_a \varphi^{\top} \hat{W}_c - \frac{k_{c2}}{N} \sum_{i=1}^N \Tilde{W}_a^{\top} \frac{G_{\phi, i}}{4 {\xyr \vartheta}_i^2} \hat{W}_a \varphi_i^{\top}  \hat{W}_c \nonumber\\
    = &  -\bar{k}_a \Tilde{W}_a^{\top}\Tilde{W}_a + k_{a1}\Tilde{W}_a^{\top} \Tilde{W}_c + k_{a2} \Tilde{W}_a^{\top} W \nonumber\\
    & -  k_{c1} \Tilde{W}_a^{\top} \frac{G_{\phi}}{4 {\xyr \vartheta}^2} W \varphi^{\top} W  + k_{c1} \Tilde{W}_a^{\top} \frac{G_{\phi}}{4 {\xyr \vartheta}^2} \Tilde{W}_a \varphi^{\top} W \nonumber\\
    & + k_{c1} \Tilde{W}_a^{\top} \frac{G_{\phi}}{4 {\xyr \vartheta}^2} W \varphi^{\top} \Tilde{W}_c -  \frac{k_{c2}}{N} \sum_{i=1}^N \Tilde{W}_a^{\top} \frac{G_{\phi, i}}{4 {\xyr \vartheta}_i^2} W \varphi_i^{\top} W  \nonumber\\
    & - \frac{1}{4} \Tilde{W}_c^{\top} \Big( k_{c1} \frac{\varphi}{{\xyr \vartheta}^2} \Tilde{W}_a^{\top} G_{\phi} \Tilde{W}_a + \frac{k_{c2}}{N}\sum_{i=1}^N \frac{\varphi_i}{{\xyr \vartheta}^2_i} \Tilde{W}_a^{\top} G_{\phi, i} \Tilde{W}_a \Big) \nonumber\\
    & + \frac{k_{c2}}{N} \sum_{i=1}^N \Tilde{W}_a^{\top} \frac{G_{\phi, i}}{4 {\xyr \vartheta}_i^2} (\Tilde{W}_a \varphi_i^{\top} W + W \varphi_i^{\top} \Tilde{W}_c),
\end{align}
where $\bar{k}_a = k_{a1}+k_{a2}$.

Combining \eqref{ea.9} and \eqref{ea.12} yields 
\begin{align} \label{ea.13}
    \dot{L}_c + \dot{L}_a \leq & -k_{c2} \bar{\lambda}_c \lVert \Tilde{W}_c\rVert^2 -(\bar{k}_a - \mu_a) \lVert \Tilde{W}_a\rVert^2 \nonumber\\
    & + \frac{3\sqrt{3}}{16} \bar{k}_c \bar{\varepsilon}_H \lVert \Tilde{W}_c \rVert   + (k_{a1}+ \mu_a) \lVert \Tilde{W}_a \rVert \lVert \Tilde{W}_c \rVert \nonumber\\
    &  + ( k_{a2} + \mu_a) \bar{W} \lVert \Tilde{W}_a\rVert + k_{c1} \frac{3\sqrt{3}}{16} \lVert \Tilde{W}_c \rVert \lVert \xi \rVert
\end{align}
where $\mu_a  = \frac{3\sqrt{3}}{64} \bar{k}_c \bar{W} \bar{G}_{\phi}$.
Define $\bar{\mu}_a =\mu_a +  \frac{\gamma}{4} \bar{G}_{\phi}$.
Then
\begin{align} \label{ea.14}
    \dot{P} = & \dot{L}_x + \dot{L}_c + \dot{L}_a + \dot{L}_u  \nonumber \\
    \leq 
    & -(1 - \gamma)x^{\top} Q x - YK_s^2 -k_{c2} \bar{\lambda}_c \lVert \Tilde{W}_c\rVert^2 -K_f \lVert \Tilde{u}^f \rVert^2 \nonumber  \\
    & - (\bar{k}_a -\bar{\mu}_a) \lVert \Tilde{W}_a\rVert^2  + l_c \lVert \Tilde{W}_c \rVert  + l_a \lVert \Tilde{W}_a \rVert + l_k \lVert K_s \rVert
    \nonumber \\
    &  + l_u \lVert \Tilde{u}^f \rVert + l_{ac} \lVert \Tilde{W}_c \rVert \lVert \Tilde{W}_a \rVert + l_{kc} \lVert \Tilde{W}_c \rVert \lVert K_s \rVert + \frac{\gamma}{4}  \bar{G}_{\phi} \bar{W}^2 \nonumber \\
    &   + l_{fc} \lVert \Tilde{W}_c \rVert \lVert \Tilde{u}^f \rVert  + l_2  \bar{K}_s \bar{G}_R \lVert \nabla \mathcal{B} \rVert_{\infty} + \bar{\varepsilon}_1 ,
\end{align}
%
where 
\begin{align*}
    l_c  = &~\frac{3\sqrt{3}}{16}\bar{k}_c \bar{\varepsilon}_H + \frac{3\sqrt{3}}{16}k_{c1} \bar{\sigma}(R) \left(\eta_1^2 + (\bar{\Tilde{u}}^{f})^2 \right) \\
     & + \frac{3\sqrt{3}}{16}k_{c1} \bar{K}_s^2 \bar{G}_R \lVert \nabla \mathcal{B}\rVert_{\infty}^2, \\
    l_a  = &~\frac{l_1}{2} + (k_{a2} + \mu_a) \bar{W}, ~~~l_k  = l_2 \bar{G}_R \lVert \nabla \mathcal{B} \rVert_{\infty},\\
    l_u  = &~l_2 \bar{g} + \lVert \dot{u}^f(t) \rVert_\infty,~~l_{kc}  = \frac{3\sqrt{3}}{8}k_{c1}\eta_1 \bar{g} \lVert \nabla \mathcal{B} \rVert_{\infty},\\
    l_{ac}  = &~k_{a1} + \mu_a + \frac{3\sqrt{3}}{16}k_{c1} \left( \bar{K}_s \bar{G}_R \lVert \nabla \mathcal{B} \rVert_{\infty} \right. \\
    &\left. +  (\eta_1 + \bar{\Tilde{u}}^f)  \bar{g}  \right) \bar{\phi}_\nabla,  \\
    l_{fc}  = &~\frac{3\sqrt{3}}{8}k_{c1}(\eta_1 \bar{\sigma}(R) + \bar{g} \bar{K}_s \lVert \nabla \mathcal{B} \rVert_{\infty}).
\end{align*}

Using Young's Inequality, \eqref{ea.14} can be further written as
\begin{align} \label{ea.15}
    \dot{P} \leq & -(1 - \gamma)x^{\top} Q x - \frac{Y}{4}K_s^2 - \frac{k_{c2} \bar{\lambda}_c}{4} \lVert \Tilde{W}_c\rVert^2  \nonumber  \\
    & - \frac{\bar{k}_a}{4} \lVert \Tilde{W}_a\rVert^2 -\frac{K_f}{4} \lVert \Tilde{u}^f \rVert^2 - y M y^{\top} + \varsigma,
\end{align}
where $y = \col\bm\left(\lVert \Tilde{W}_c \rVert, \lVert  \Tilde{W}_a \rVert, \lVert \Tilde{u}^f \rVert, \lVert K_s \rVert \bm\right)$, $\varsigma = \frac{l_c^2}{2k_{c2} \lambda_c} + \frac{l_a^2}{2\bar{k}_a} + \frac{l_k^2}{2Y} + \frac{L_u^2}{2K_f} + \frac{\gamma}{4}  \bar{G}_{\phi} \bar{W}^2 + l_2  \bar{K}_s \bar{G}_R \lVert \nabla \mathcal{B} \rVert_{\infty}  + \bar{\varepsilon}_1$ and 
\begin{align*}
    M = 
    \begin{bmatrix}
        \frac{k_{c2} \bar{\lambda}_c}{4} & -\frac{l_{ac}}{2} & -\frac{l_{kc}}{2} & -\frac{l_{fc}}{2} \\
        -\frac{l_{ac}}{2} & \frac{\bar{k}_a}{4} - \bar{\mu}_a & 0 & 0 \\
        -\frac{l_{kc}}{2} & 0 & \frac{Y}{4} & 0\\
        -\frac{l_{fc}}{2} & 0 & 0 & \frac{K_f}{4}
    \end{bmatrix}.
\end{align*}
By designing $k_{c1}, k_{c2}, k_{a1}, k_{a2}, \beta, Y, \gamma, \omega, \Gamma(0)$ such that 
\begin{align*}
 l_{fc} &\leq \frac{K_f}{2},\quad 2\mu_a + l_{ac} \leq \frac{\bar{k}_a}{2},\\
l_{kc} &\leq \frac{Y}{2} \quad \textnormal{and}\quad  l_{ac} + l_{kc} + l_{fc} \leq  \frac{k_{c2} \bar{\lambda}_c}{2}
\end{align*}
 one has $\underline{\sigma}(M) \geq 0$ according to Geršgorin circle criterion.
Then 
\begin{align} \label{ea.16}
    \dot{P} \leq & -(1 - \gamma)x^{\top} Q x - \frac{Y}{4}K_s^2 - \frac{k_{c2} \bar{\lambda}_c}{4} \lVert \Tilde{W}_c\rVert^2 \nonumber \\
    & - \frac{\bar{k}_a}{4} \lVert \Tilde{W}_a\rVert^2 -\frac{K_f}{4} \lVert \Tilde{u}^f \rVert^2 + \varsigma \nonumber  \\
    \leq & - \underbrace{\operatorname{min}\Big \{1 - \gamma, \frac{Y}{4}, \frac{k_{c2} \bar{\lambda}_c}{4},  \frac{\bar{k}_a}{4}, \frac{K_f}{4}\Big\}}_{\kappa>0}\lVert z \rVert^2 + \varsigma.
\end{align}
%
The time derivative of \eqref{ea.1} is upper bounded as follows
\begin{align*}
    \dot{P} \leq - \frac{\kappa}{2} \lVert z \rVert^2, & &  \forall \lVert z \rVert  \geq \sqrt{\frac{2\varsigma}{\kappa}}.
\end{align*}
According to Theorem 4.18 in \cite{khalil2002nonlinear}, one can conclude that the state trajectories $x$, the NN weights estimation errors $\Tilde{W}_c$, $\Tilde{W}_a$, the adaptive safeguard gain $K_s$ and the fault observation error $\Tilde{u}^f$ are uniformly ultimately bounded. 
\end{proof}

\section{Simulation examples} \label{sec_simulation}
In this section, we study the performance of the proposed learning method relative to other existing safe RL methods.
%
\begin{example}
%
Consider an inverted pendulum system \cite{wabersich2023data}  
\begin{align*}
\begin{bmatrix}
    \dot{x}_1\\
    \dot{x}_2
\end{bmatrix} =
\begin{bmatrix}
    {\xywr x_2}\\({g_{\theta}}/{l_{\theta}})\sin(x_1)
\end{bmatrix} 
+
\underbrace{\begin{bmatrix}
    0\\{1}/{m_{\theta}l_{\theta}^2}
\end{bmatrix} 
}_{{\xyr g(x_1,x_2)}}
(u+ u^f),
\end{align*}
where $x_1=\theta$ is the angle, $x_2= \dot{\theta}$ is the angular velocity, $u$ and $u^f$ are the input torque and the disturbance/actuator fault, respectively; $g_\theta = [0,{1}/{m_{\theta}l_{\theta}^2}]^\top$.
%
The initial states are $\theta(0) = 0.5 \si{\degree}$ and $\dot{\theta}(0) = \SI{10}{\degree / \second}$.
The mass is   $m_{\theta} = 2 \; \si{kg}$,  the length is $l_{\theta} = 1 \; \si{m}$, and the gravitational acceleration is $g_{\theta} = 10 \; \si{m/s^2}$.
The state constraints are $\theta \leq 0.8 \si{\degree}$ and $\dot{\theta} \geq \SI{-2}{\degree /\second}$.
The relative degree of the constraint on angle is two, and the relative degree of the constraint on angular velocity is one.
For $\dot{\theta} \geq -2 \si{\degree / \second}$, define $\psi_{\dot{\theta},0} =\dot{\theta} + 2$ and $B_{\dot{\theta}} = \frac{1}{\psi_{\dot{\theta},0} }$; for $\theta \leq 0.8\si{\degree}$, define $\psi_{\theta,0} =0.8-\theta$, $\psi_{\theta,1} =-\dot{\theta} + 100(0.8 - \theta)$ and $B_{\theta} = \frac{1}{\psi_{\theta,1}}$.
{\xyr Then $\mathcal{L}_g B_{\dot{\theta}} = -\frac{1}{m_\theta l_\theta^2\psi_{\dot{\theta},0}^2}\neq 0$ for all $\psi_{\dot{\theta},0}>0$ and $\mathcal{L}_g B_{\theta} = \frac{1}{m_\theta l_\theta^2\psi_{\theta,1}^2}\neq 0$ for all $\psi_{\theta,1} > 0$, which implies that $B_{\dot{\theta}}$ and $B_{\theta}$ are HO-RCBFs.}
The design parameters used in this example are: $Q = I,~ R = 1,~\phi(\theta, \dot{\theta}) = [\theta^2, \theta\dot{\theta},\dot{\theta}^2]^\top,~k_{c1} = 0.1, ~k_{c2} = 1, ~k_{a1} = 100, ~k_{a2} = 1, ~\beta = 0.1,~\hat{W}_c(0) = \hat{W}_a(0) = [40, 120, 30], ~\Gamma(0) = 1000I,~\omega(\theta, \dot{\theta}) = 20\,\dot{\theta},~\textnormal{and}~K_s = 1$.

In this example, we will compare our proposed method with classical RL \cite{vamvoudakis2010online}, penalty-based RL \cite{marvi2021safe} and safety-filter-based RL \cite{peng2023design}. 
The performance weights are selected as $Q = I, R = 1$.
The activation function is chosen as $\phi(\theta, \dot{\theta}) = \begin{bmatrix} \theta^2 &\theta\dot{\theta}& \dot{\theta}^2\end{bmatrix}$.
The learning gains are set as $k_{c1} = 0.1, k_{c2} = 1, k_{a1} = 100, k_{a2} = 1, \beta = 0.1$, and the initial weights are designed as $\hat{W}_c(0) = \hat{W}_a(0) = [40, 120, 30], \Gamma(0) = 1000I$.
The user-defined function in the observer is selected as $\omega(\theta, \dot{\theta}) = 20\,\dot{\theta}$.
The safeguard gain is selected as $K_s = 1$.

\begin{figure*}[htb]
    \centering
    \begin{subfigure}{0.45\textwidth}
        \centering
        \includegraphics[width=1\textwidth, height=0.3\textheight]{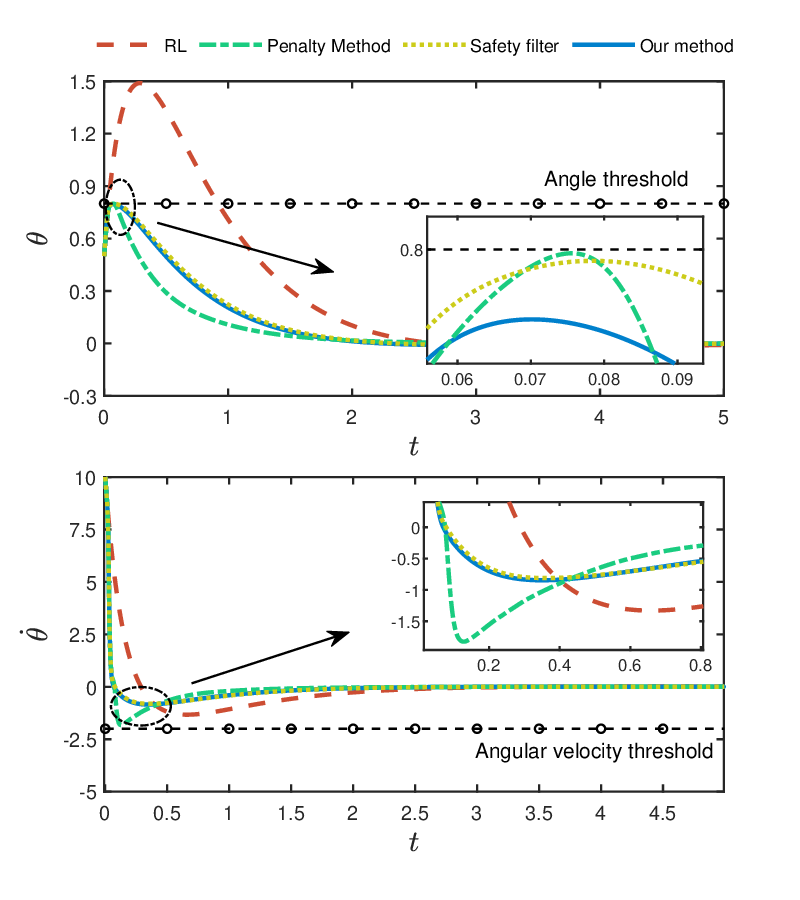}
        \caption{}
        \label{fig2.a}
    \end{subfigure}  
    \quad
    \begin{subfigure}{0.45\textwidth}
        \centering
        \includegraphics[width=1\textwidth, height=0.3\textheight]{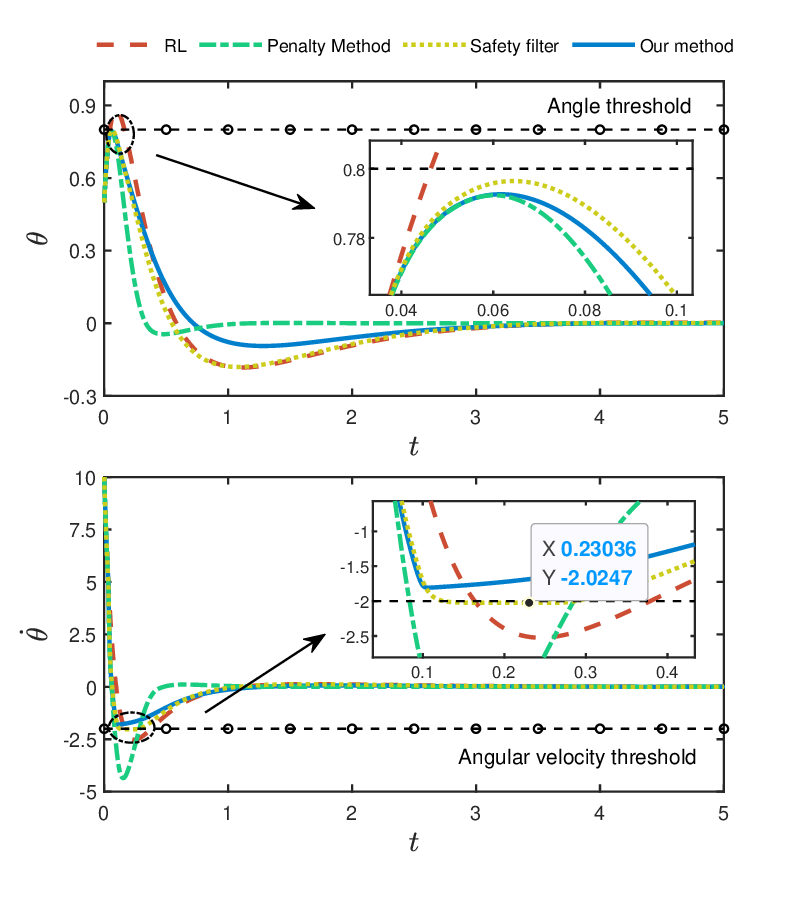}
        \caption{}
        \label{fig2.b}
        \end{subfigure}  
    \caption{{\xywv Trajectories of the angle and angular velocity under four algorithms, i.e., classical RL, RL with penalty term, RL with safety filter and our method: (a) without disturbance/actuator fault; (b) with disturbance/actuator fault} }
    \label{fig1:combined}
\end{figure*}

%
%

%
First, let $u^f(t) = 0$. 
As shown in Fig.\ref{fig2.a}, the classical RL policy will violate the constraint on angle, while all the other three safe learning methods can perform the control task without violating any state constraints.
Then, let the disturbance $u^f(t)= -5 + 0.01 \sin(t) + 0.03 \cos(t) + 0.05 \sin(2t) + 0.04 \cos(2t)$.
As shown in Fig.\ref{fig2.b}, the classical RL policy can  stabilize the system under unknown disturbance/actuator fault, but the angle constraint is violated during the learning process.
The proposed method in this study can complete the control task while guaranteeing the safety of the system.
When the safety filter is introduced into the RL method, the velocity constraint is violated due to $u^f$. 
The control policy learned from the RL method with a CBF-based penalty term cannot keep the velocity trajectory inside the safe region, which may be due to the choice of the penalty term, and the trade-off between the angle constraint and velocity constraint.
\end{example}
%
%

%


\begin{example}


In this example, the proposed method is shown to handle barrier functions with different relative degrees.  Consider a mobile robot described by a double integrator $\dot{p} = v$ and $\dot{v} = u$, where $p = [p_1,p_2]^{\top}, v = [v_1, v_2]^{\top}$ and $u \in \mathbb{R}^2$ represent the position, the velocity and the input, respectively.
%
The task of the mobile robot is to reach the origin without leaving a circular area centered at the origin while avoiding obstacles and maintaining a defined range of velocity.

Obstacle areas $O_i$, $i =  1,2,\cdots, N$, are placed randomly inside the circular area, and the collision avoidance constraint for $i$th obstacle area is given by $\lVert p - c_i\rVert^2 \geq r_i^2$, where $c_i$ and $r_i$ are the center and radius of the $i$th obstacle, respectively.
The circular area constraint is described as $\lVert p \rVert \leq r^2$, where $r$ is the radius of the area.
The constraints on the robot velocity are set to $v_1,v_2\in[-0.8,0.8] \si{m / \second}$.

For the velocity constraint with relative degree one, define $B_{v_j} = \frac{1}{h_{v_{j, max}}} + \frac{1}{h_{v_{j, min}}},~ \textnormal{for}~ j = 1, 2$, where $h_{v_{j, max}}  = 0.8 - v_j $ and $h_{v_{j, min}}  = v_j +0.8$.
For collision avoidance constraint with relative degree two, define $\psi_{o_i, 0} = h_{o_i}, ~\psi_{o_i, 1} = \dot{\psi}_{o_i, 0} + h_{o_i}$ and $B_{o_i} = \frac{1}{\psi_{o_i, 1}},\  \textnormal{for}\ i = 1,2,\cdots, N, $ where $h_{o_i} = \lVert p - c_i \rVert^2 -  r_i^2$.
Similarly, for the area constraint, define $\psi_{area, 0} = h_{area}, \quad \psi_{area, 1} = \dot{\psi}_{area, 0} + h_{area}$ and $B_{area} = \frac{1}{\psi_{area, 1}}$, where $h_{area} = r^2 - \lVert p \rVert^2$.
{\xyr Following a similar development in Example 2, we can show that $B_{v_j}, B_{o_i}$ and $B_{area}$ are HO-RCBFs.}
Define the safeguard gains for the position and velocity constraint as $K_{s,p}$ and $K_{s,v}$, respectively. 
The design parameters used in this example are: 
$Q = R = I,~Y = 500,~\gamma = 0.001,~K_{s,p}(0) = 10,~K_{s,v}(0) = 0.01,~\phi(p,v)=[ p_1^2, p_2^2, v_1^2, v_2^2, p_1p_2, p_1 v_1, p_1 v_2, p_2 v_1, p_2 v_2, v_1 v_2 ]^{\top},~k_{c1} = 0.1, k_{c2} = 1,~ k_{a1} = 100,~ k_{a2} = 1,~ \beta = 0.1,~ \Gamma(0) = 1000I$.
%




\begin{figure*}[htb]
    \centering
    \begin{subfigure}{0.45\textwidth}
        \centering
        \includegraphics[width=1\textwidth]{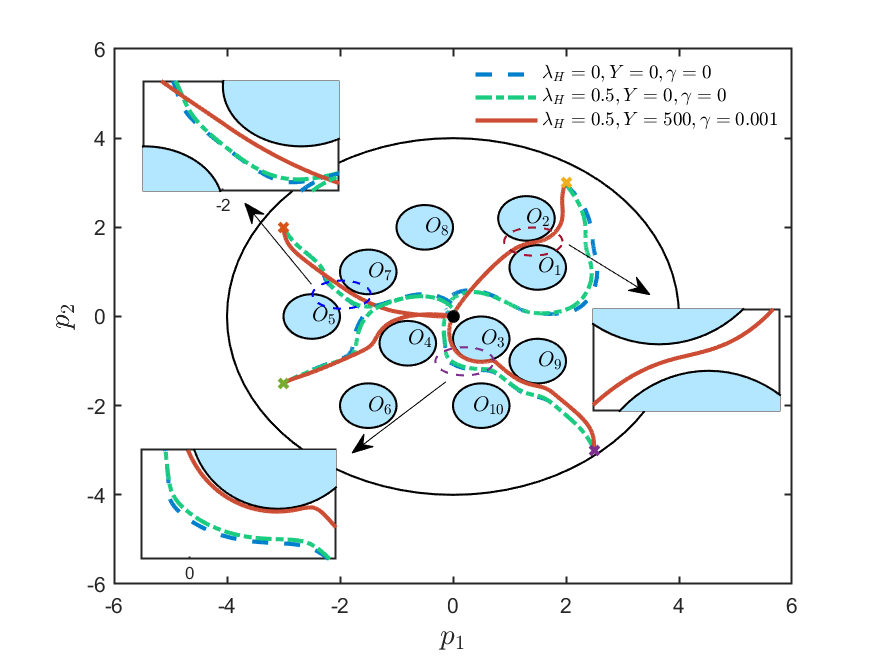}
        \caption{}
        \label{fig3.1}
    \end{subfigure} 
     \quad
    \begin{subfigure}{0.45\textwidth}
        \centering
        \includegraphics[width=1\textwidth]{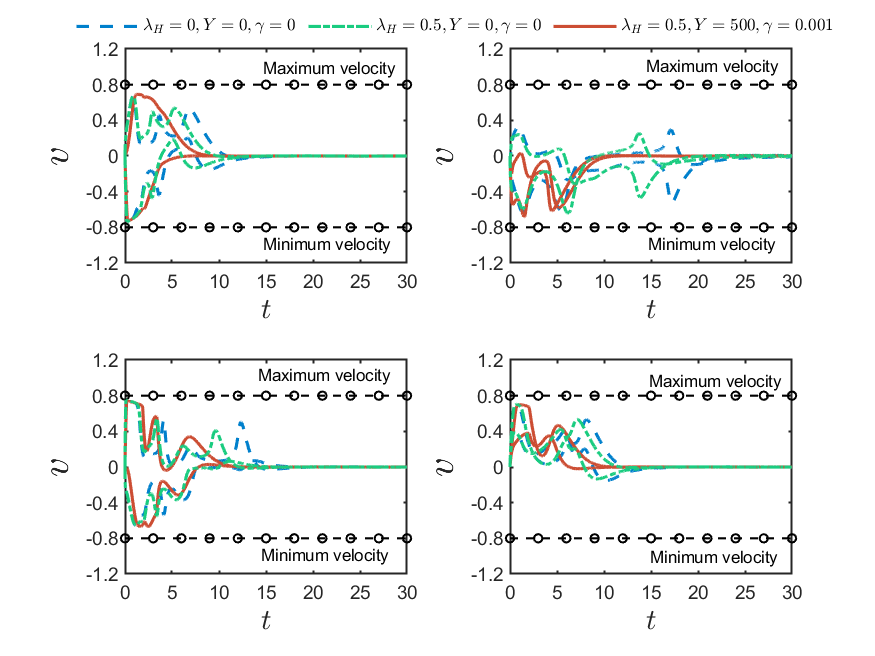}
        \caption{}
        \label{fig3.2}
        \end{subfigure}  
    \caption{Position and velocity trajectories of the mobile robot under HO-RCBF-based RL with a constant safeguard gain (blue dashed line), HO-RCBF-based RL with an adaptive safeguard gain (green dashed line and red solid line).}
    \label{fig3:combined}
\end{figure*}


\begin{figure}[htbp]
    \centering
    \includegraphics[width=0.5 \textwidth, height=0.22 \textheight]{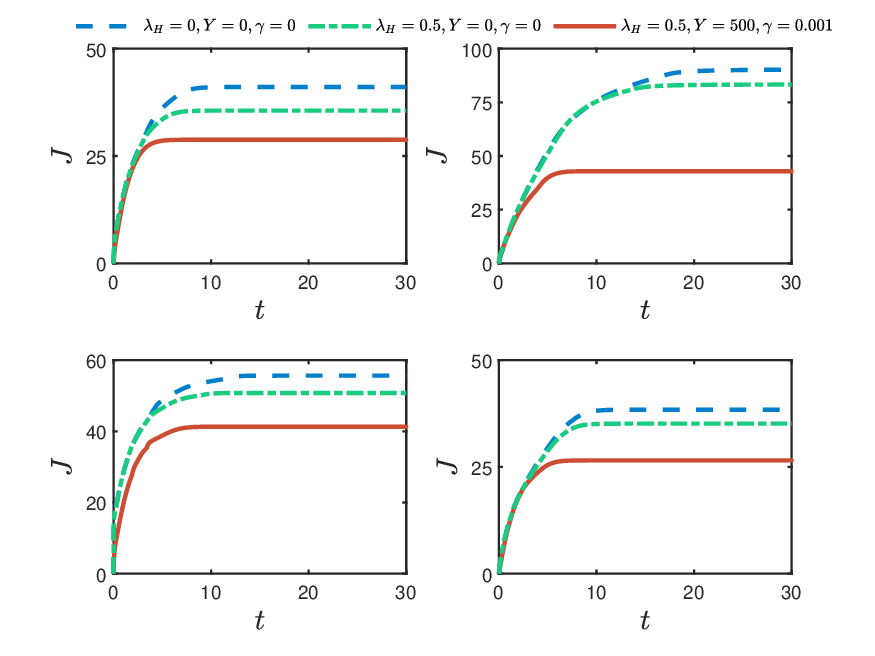}    
    \caption{Cost function of the mobile robot under both HO-RBCF-based RL (blue dashed line) and adaptive HO-RBCF-based RL algorithms (green dashed line and red solid line). 
    }  \label{fig4}  
\end{figure}

We illustrate the improved performance of the gradient manipulation technique and the adaptive mechanism in the search for the optimal trajectory. To show this, we apply the HO-RCBF-based RL with $\mu = 0, \gamma = 0, Y = 0$, the adaptive HO-RCBF-based RL with $\mu = 0.5, \gamma = 0, Y = 0$, and the adaptive HO-RCBF-based RL with $\mu = 0.5, Y = 500, \gamma = 0.001$.
For all methods, $k_{s,v}$ is fixed, and $k_{s,p}$ is updated using the adaptive mechanism.
As shown in Fig.\ref{fig3:combined}(\subref{fig3.1}) and Fig.\ref{fig3:combined}(\subref{fig3.2}), the robot using both the HO-RCBF-based RL algorithm and the adaptive HO-RCBF-based RL algorithm can reach the origin while satisfying all speed and position constraints.
However, the position trajectory of the robot using the HO-RCBF-based RL algorithm is more conservative than that the adaptive HO-RCBF-based RL algorithm (see Fig. \ref{fig3:combined}(\subref{fig3.1})).
Figure \ref{fig3:combined}(\subref{fig3.2}) additionally shows that the robot using the adaptive HO-RCBF-based RL algorithm can reach the target much faster than the HO-RCBF-based RL algorithm. 
We also consider the energy consumption $J$ required by the robot to perform its task. A higher value of $J$ indicates a greater energy consumption. 
As shown in Fig.\ref{fig4}, given different initial positions {$\col(-3, -2), \col(2, 3),\col(2.5, -3), \col(-3, -1.5)$}, the accumulated cost $J$ of the robot using the adaptive HO-RCBF-based RL algorithm is always significantly less than that of the HO-RCBF-based RL algorithm.


\end{example}

\section{Conclusion}
This paper proposed an adaptive robust safe learning framework that solve the constrained optimal control problem of nonlinear systems subject to disturbances/faults.
The relationship between performance and safety is quantified and a balance between these two objectives is achieved. An HO-RCBF is introduced, which can be integrated with RL to address constrained optimal control problems. The proposed safeguard control policy can handle safety constraints with high relative degree. Interestingly, it can be integrated into any stabilizing control law, designed without considering safety issue, to guarantee the safety of the system.  




\section*{References}
\bibliographystyle{IEEEtran}        
\bibliography{xywang}          
\vskip 1cm
\begin{IEEEbiography}[{\includegraphics[height=1.25in,keepaspectratio]{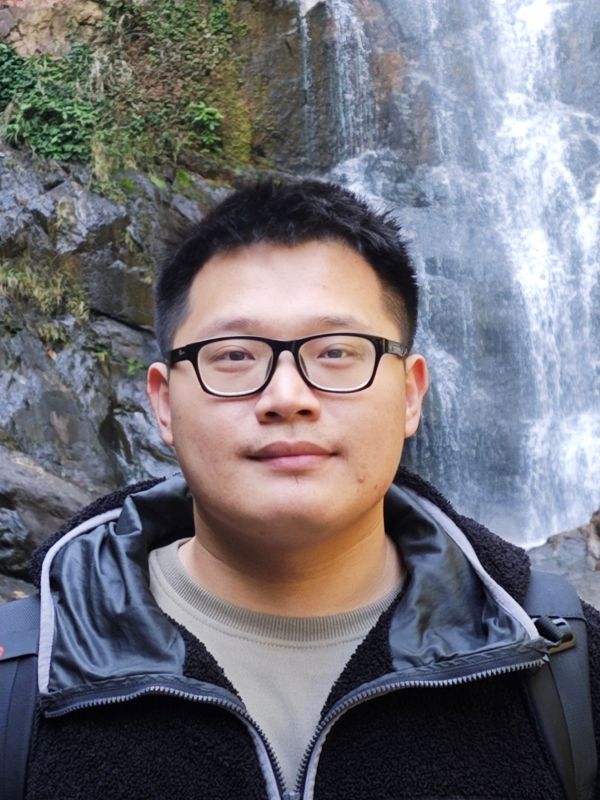}}]{Xinyang Wang} received the M.Eng. in control engineering from Harbin Institute of Technology, Harbin, China, in 2022. He is currently pursuing the Ph.D. in control science and control engineering from Harbin Institute of Technology, Shenzhen, China. His current research interests include reinforcement learning, safe learning, game theory and multi-agent systems.
\end{IEEEbiography}

\vskip 1cm

\begin{IEEEbiography}[{\includegraphics[height=1.25in,keepaspectratio]{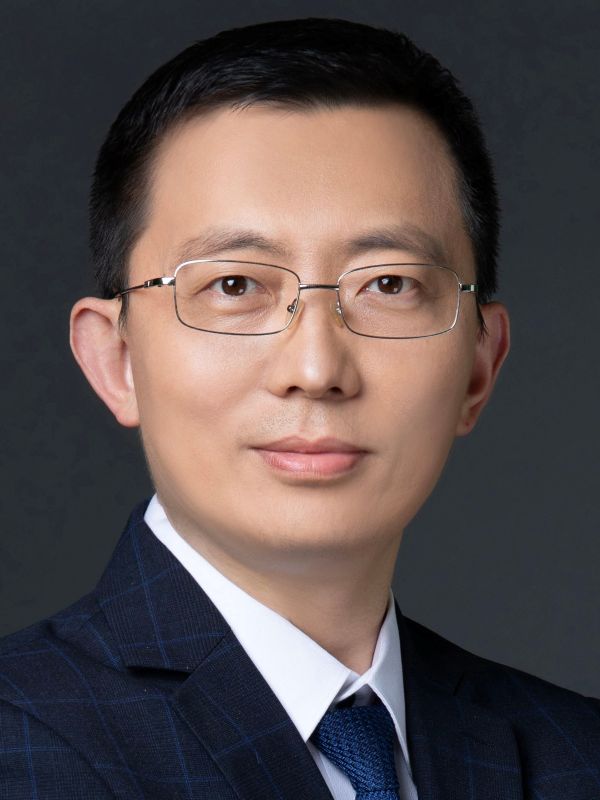}}]{Hongwei Zhang}  received 
the Ph.D. in mechanical and automation engineering from the Chinese University of Hong Kong in 2010.
Subsequently, he held postdoctoral positions with the University of Texas at Arlington, Arlington, TX, USA, and the City University of Hong Kong. From 2012 to 2020, he was with Southwest Jiaotong University, China, and then joined the Harbin Institute of Technology, Shenzhen, China, in 2020, as a Professor. His research interests are cooperative control of multiagent systems, with its applications to unmanned systems, microgrids, and active noise control.
Dr. Zhang is an Associate Editor for Neurocomputing.
\end{IEEEbiography}

\vskip 1cm

\begin{IEEEbiography}[{\includegraphics[width=1in,height=1.25in,keepaspectratio]{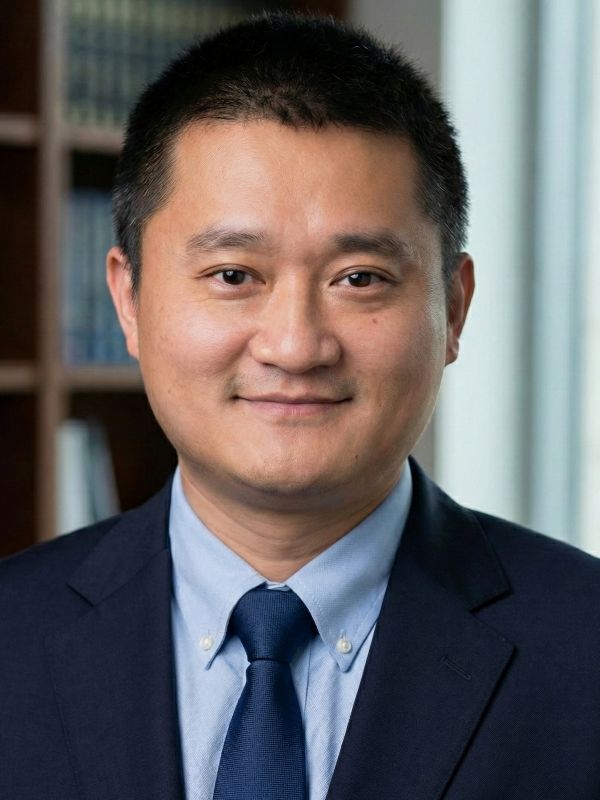}}]{Shimin Wang} received his Ph.D. from the Chinese University of Hong Kong. He is a recipient of the NSERC Postdoctoral Fellowship, the Best Poster Award at the 2024 Nonlinear Systems and Control Conference, and Best Conference Paper Awards at the 2025 IEEE International Conference on Unmanned Systems and 2018 IEEE International Conference on Information and Automation. His research bridges data science, embodied AI, and machine learning to build advanced intelligent autonomous systems. 
\end{IEEEbiography}

\vskip 1cm

\begin{IEEEbiography}[{\includegraphics[width=1in,height=1.25in,keepaspectratio]{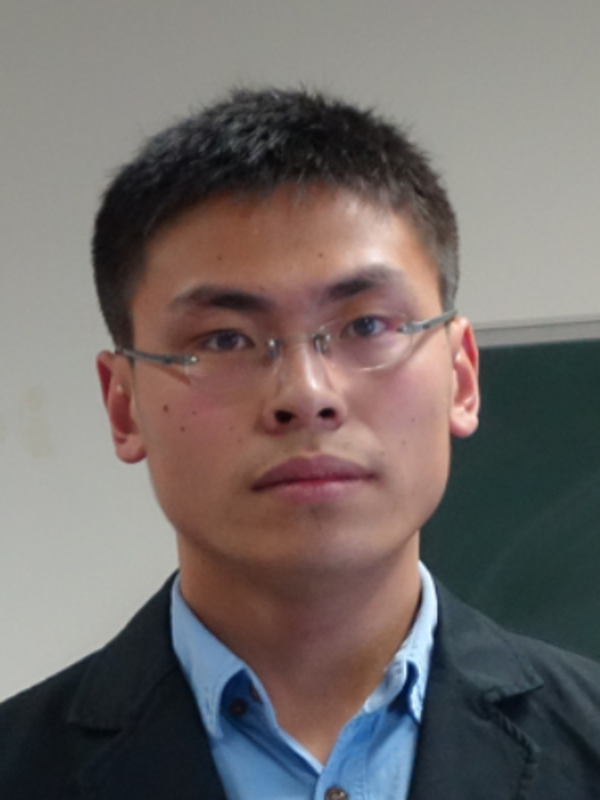}}]{Wei Xiao} is a Nanyang assistant professor at the School of Electrical and Electronic Engineering, Nanyang Technological University, Singapore, a PI at M3S, SMART, and a research affiliate with MIT CSAIL. 
He was an assistant professor at the robotic engineering department at WPI and a postdoctoral associate at Massachusetts Institute of Technology. He received a B.Sc. degree from the University of Science and Technology Beijing, China in 2013, a M.Sc. degree from the Chinese Academy of Sciences (Institute of Automation), China in 2016, and a Ph.D. degree from the Boston University, Brookline, MA, USA in 2021.
His research interests include control theory and machine learning, with particular emphasis on robotics and traffic control. He received an Outstanding Student Paper Award at the 2020 IEEE Conference on Decision and Control.
\end{IEEEbiography}

\vskip 1cm

\begin{IEEEbiography}[{\includegraphics[height=1.25in, keepaspectratio]{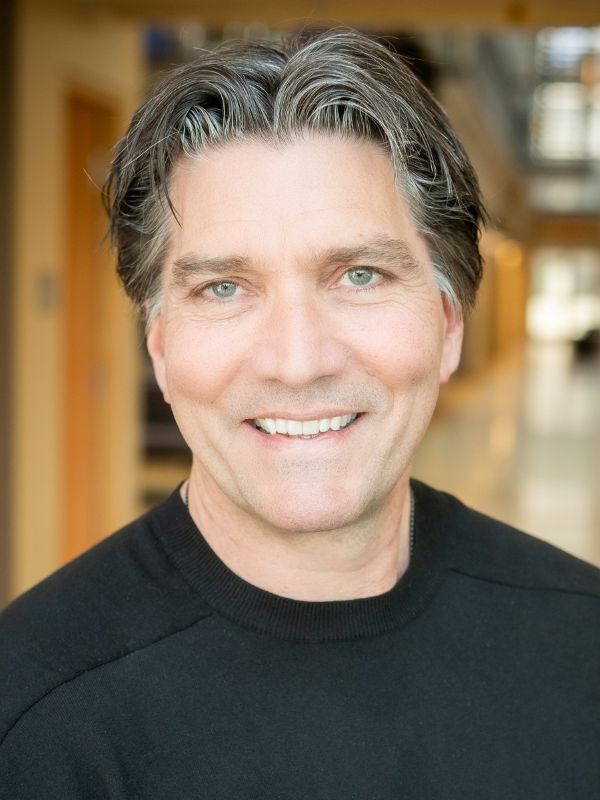}}]{Martin Guay} received a Ph.D. from Queen’s University, Kingston, ON, Canada in 1996. He is currently a Professor in the Department of Chemical Engineering at Queen’s University. His current research interests include nonlinear control systems, especially extremum-seeking control, nonlinear model predictive control, adaptive estimation and control, and geometric control. 
He was a recipient of the Syncrude Innovation Award, the D. G. Fisher from the Canadian Society of Chemical Engineers, and the Premier Research Excellence Award. He is a Senior Editor of IEEE Transactions on Automatic Control. He is the Editor-in-Chief of the Journal of Process Control. He is also an Associate Editor for Automatica and the Canadian Journal of Chemical Engineering.
 \end{IEEEbiography}

\end{document}